\documentclass{llncs}
\let\origvec\vec

\let\vec\origvec
\usepackage{amsmath}
\usepackage{amssymb}
\usepackage{framed}
\usepackage{graphicx}
\usepackage{etoolbox}

\usepackage[matrix,frame,arrow]{xy}
\usepackage{amsmath}
\newcommand{\qw}[1][-1]{\ar @{-} [0,#1]}
\newcommand{\multigate}[2]{*+<1em,.9em>{\hphantom{#2}} \qw \POS[0,0].[#1,0];p !C *{#2},p \save+LU;+RU **\dir{-}\restore\save+RU;+RD **\dir{-}\restore\save+RD;+LD **\dir{-}\restore\save+LD;+LU **\dir{-}\restore}
\newcommand{\ghost}[1]{*+<1em,.9em>{\hphantom{#1}} \qw}

\newcommand{\rstick}[1]{*!L!<-.5em,0em>=<0em>{#1}}
\newcommand{\lstick}[1]{*!R!<.5em,0em>=<0em>{#1}}

\newcommand{\Qcircuit}{\xymatrix @*=<0em>}

\spnewtheorem{protocol}[theorem]{Protocol}{\bfseries}{\itshape}

\let\emptyset\varnothing

\usepackage{tikz}
\usepackage{tikzpeople}
\usepackage{stackengine}
\usepackage{relsize}
\usepackage{setspace}
\usepackage{pdfpages}
\usepackage{marvosym}

\title{
Non-Locality in Interactive Proofs}

\author{
Claude Cr\'epeau \inst{1}\! \thanks{Supported in part by FRQNT (INTRIQ) and NSERC (CryptoWorks21 and Discovery grant program).}
\and 
Nan Yang\inst{2}\! \thanks{Supported in part by Professors V\'aclav~Chv\'atal, Jeremy~Clark, Claude~Cr\'epeau, and David~Ford.}
}

\institute{
McGill University, Montr\'eal, Qu\'ebec, Canada.
{crepeau@cs.mcgill.ca}
\and
Concordia University, Montreal, Quebec, Canada.
{na\_yan@encs.concordia.ca}
}

\begin{document}

\maketitle

\begin{abstract}


In multi-prover interactive proofs (MIPs), the verifier is usually non-adaptive. This stems from an implicit problem which we call ``contamination'' by the verifier. We make explicit the verifier contamination problem, and identify a solution by constructing a generalization of the MIP model. This new model quantifies non-locality as a new dimension in the characterization of MIPs. A new property of zero-knowledge emerges naturally as a result by also quantifying the non-locality of the simulator.
\end{abstract}

\section{Introduction}

An \emph{interactive proof} is a dialog between two parties: a polynomial-time \emph{verifier} and an all-powerful \emph{prover} \cite{GoldwasserMiRa89,Babai85}. They agree ahead of time on some language $L$ and a string $x$. The prover wishes to convince the verifier that $x \in L$. If this is true, the prover should succeed almost all the time; if not, the prover should fail almost all the time. This is a generalization of the complexity class $\mathbf{NP}$, except instead of simply being handed a polynomial-sized witness, the verifier is allowed to quiz the prover. The set of languages that admit an interactive proof is called $\mathbf{IP}$.

The \emph{multi-prover} model was introduced in \cite{BGKW88}. This model consists of multiple, non-commu\-nicating\footnote{The precise meaning of these words shall become a lot clearer throughout the rest of this paper.} provers talking to a single verifier. The inspiration for this model was that of a detective interrogating a number of suspects, each of whom is isolated in a separate room. The suspects may share a strategy before being separated, but once the interrogation begins they are no longer able to talk to one another. The main motivation for studying this model was to remove the complexity assumptions used in the commitment schemes. We will abbreviate ``multi-prover interactive proof'' as MIP and the set of languages which can be accepted by MIPs as the boldface $\mathbf{MIP}$.

Implicit in the definition of the multi-prover model (in the original \cite{BGKW88}) is that the provers are \emph{local}. That is, not only do the provers not communicate, but they are not correlated in any way beyond sharing random bits.

An interactive proof is \emph{zero-knowledge} if the verifier learns nothing except the truth of ``$x \in L$''. This is usually defined by saying that a \emph{distinguisher} is unable to tell apart a real conversation between the prover and the verifier, and one which is generated by a lone polynomial-time \emph{simulator}. We will denote sets of zero-knowledge interactive proofs with a $\mathbf{ZK}$ bold prefix.

From a complexity perspective, the zero-knowledge aspect of interactive proofs is characterized by $\mathbf{IP} = \mathbf{ZKIP} = \mathbf{PSPACE}$ for single-prover IPs (\cite{Shamir:1992:IP,ImpagliazzoY87,Ben-Or:1988}), and $\mathbf{MIP} = \mathbf{ZKMIP} = \mathbf{NEXP}$ for multi-prover IPs (\cite{BGKW88,Fortnow:1994:PMI:194527.194556,BFL90,KILIAN89,DBLP:conf/stoc/FeigeK94,DBLP:conf/crypto/DworkFKNS92,DBLP:journals/siamcomp/FeigeK00}). The (conjectured) necessity of complexity assumptions for zero-knowledge in the single-prover case was the initial motivation for the multi-prover model.

However, there is a relationship between \emph{non}-locality and zero-knowledge which remains unexplored. Let us call this the cryptographic characterization (or perspective) of ZKMIPs.

\subsection{A Cryptographic Perspective}


The foundation of zero-knowledge is the idea of a \emph{simulator}, a machine with no more power than the verifier, which can pretend to be all-powerful provers. Obviously, this simulator cannot accomplish this task without some kind of \emph{advantage} -- independent of knowledge -- that must be provided. In single-prover zero-knowledge proofs, this advantage can be in the form of the ability to {\em rewind} computation, to discard failed simulations, or knowledge of a trapdoor in the commitment scheme. In multi-prover zero-knowledge proofs, the advantage in existing literature can be summed up as \emph{signaling}: the simulator, acting as several provers, knows secrets which real provers, in a real instance of the protocol, would not. This is then used to produce the simulation.

This signaling advantage of existing ZKMIP simulators is unnecessarily strong in the sense that if we were to require the transcript to come from multiple, non-communicating simulators (as we do with provers in real instances), then existing simulation strategies would fail (as they would require the simulators to communicate), whereas we have discovered that there exist simulation strategies which do not require communication. Instead, we only require some level of \emph{non-local correlation} between the simulators. The exact level of correlation required is a heretofore uncharacterized dimension in interactive proofs.

In order to build the framework necessary to express and characterize this dimension, we begin with an implicit problem in the existing MIP literature.

\subsection{Implicit Problem / Ad Hoc Solution}

There is an implicit problem in what we call the ``standard'' MIP model (one verifier talking to a number of provers) in the existing literature. As a lead-up to describing this problem, we invite the readers to consider the following ridiculous two-prover protocol:
\pagebreak

\begin{quote}
\rule{\linewidth}{1pt}
\begin{protocol} ( Ridiculous Protocol ) \end{protocol}

\begin{enumerate}
\item Verifier sends Prover 1 a random string $S$.
\item Prover 1 replies with a string $T$.
\item Verifier sends Prover 2 the string $T$.
\item Prover 2 replies with a string $S'$.
\item Verifier accepts if $S = S'$.
\end{enumerate}

\rule{\linewidth}{1pt}
\end{quote}

Suppose that we claim the following ridiculous theorem:

\begin{theorem} (Ridiculous Theorem) The probability that the verifier accepts in the Ridiculous Protocol is exponentially small. \end{theorem}

\begin{proof} (Ridiculous Proof) By the definition of MIPs, the provers cannot communicate. If Prover 2 can output an $S'$ that is the same as the uniformly random $S$ that only Prover 1 knows, then they must have communicated. Contradiction. \qed \end{proof}

The reader is astute in pointing out that steps 2 and 3 of the Ridiculous Protocol clearly show that the verifier is helping the provers by relaying the very answer it is supposed to keep secret. This is the implicit problem, exaggerated.

We will call this implicit problem ``contamination'' by the verifier. For example, a verifier talking to one prover \emph{and then} talking to another prover risks unwittingly helping the provers (up to) signal. However, the most important (and the most subtle) of those contaminations are ones where the verifier helps the provers perform a \emph{no}-signaling correlation; examples of this can be found in the following section, and also in \cite{CSST11}.

The ad hoc solution in existing literature is to cripple the verifier so that it would not do this (and much more). The verifier in existing literature is assumed to be (or constructed to be) \emph{non-adaptive}. That is, the verifier essentially chooses the questions ahead of time. This circumvents the problem of contamination.

However, this is overkill. We can address the problem of contamination without requiring the verifier to be non-adaptive. We do so by constructing a multi-prover, multi-verifier model which we shall call \emph{locality-explicit} multi-prover interactive proofs (LE-MIP). MIPs in this form have prover-verifier pairs who are talking, but no communication \emph{between} any of the pairs. At the end of a locality-explicit protocol, a special, read-only verifier accepts or rejects.

Locality-explicit protocols do not have to worry about contamination by the verifier, therefore they do not need to be non-adaptive. We will show later that LE-MIPs can be generalized to account for non-locally augmented provers without resorting to non-adaptive verifiers.

This new model offers the following advantages:



\begin{enumerate}
\item The provers and verifiers are guaranteed to be local (i.e., a very strong notion of no-commu\-nicating), if desired.
\item Any non-local resources of provers and verifiers are made explicit.
\item It is possible to enforce ``honest non-locality'' on the provers by having the \emph{verifier} provide them with non-local resources. Our model makes this explicit.
\end{enumerate}


The new characterization of ZKMIPs emerges as we naturally extend zero-knowledge to LE-MIPs, by making explicit the non-local resources of the (multiple) simulators.

\subsection{Our Contributions}


\begin{itemize}

\item We explain the aforementioned implicit problem with the standard (single-verifier) MIP model (section \ref{SEC:Standard}).

\item We describe the locality-explicit model and justify its definition by expanding on its advantages over the standard model (section \ref{SEC:CCMIP}).


\item We show that, in the LE-MIP model, a new, stronger property of zero-knowledge naturally emerges (section \ref{SEC:NEWZK}).

\item We describe a protocol which is local-verifier, local-prover and zero-knowledge which accepts oracle-3-SAT, achieving zero-knowledge without needing the provers to authenticate any messages, and prove its security (section \ref{SEC:NEXP}).

\item We describe how to simulate the above protocol with simulators which have only a specific no-signaling advantage (section \ref{PofS}).

\end{itemize}

\section{Previous Work}


The early work by Ben-Or, Goldwasser, Kilian and Wigderson asserting that $\mathbf{ZKMIP}=\mathbf{MIP}$ from \cite{BGKW88} and \cite{KILIAN89} use multi-round protocols and their (honest) verifiers are inherently signalling.
This is precisely why we address the situation in this work. Proving soundness is quite subtle in this case because the provers could use the (signalling) verifier to break binding of the commitments. In particular, soundness will not be valid if the protocol is composed concurrently with other executions of itself or even used as a sub-routine.
In recent conversations with Kilian \cite{KILIAN18}, we have learned that controlling the impact of this \emph{signaling} (via the verifier) has been a concern since the early days of MIPs.
The protocols as they are might be sound but it is not fully proven anywhere in writing.
However, it is also clear that no considerations had been given to the fact that general non-local correlations are possible via the verifier. If soundness rests on the binding property of a commitment scheme (such as those zero-knowledge proofs) and this binding property rests on the inability to achieve a certain non-local correlation then impossibility to achieve this correlation via the verifier must be demonstrated.
It is not done or hinted in these papers.

The multi-round issue we address may seem trivial because it is a known fact that multi-round MIPs may be reduced to a single round using techniques
of Lapidot-Shamir \cite{DBLP:conf/focs/LapidotS91} and Feige-Lovasz \cite{Feige:1992:TOP:129712.129783}.
Nevertheless, if interested in \emph{zero-knowledge} MIPs, commitment schemes are generally used to obtain the zero-knowledge property and thus the single-round structure is lost in the process.
Although single-round protocols bypass verifier's non-local contamination problems we describe in this work, converting multi-round protocols into single-round ones is highly inefficient and complex.
Preserving zero-knowledge while achieving single-round has turned out to be a major challenge.
Practically, keeping a multi-round protocol's structure, using only commitments to achieve zero-knowledge is very appealing.

In \cite{DBLP:conf/focs/LapidotS91}, Lapidot-Shamir proposed a parallel ZKMIP for $\mathbf{NEXP}$, but they removed the zero-knowledge claim
in the journal version \cite{DBLP:journals/jcss/LapidotS97} of their work without any explanation as of why.
Feige and Kilian \cite{DBLP:conf/stoc/FeigeK94} were the last ones to follow this approach combining techniques drawn from Lapidot-Shamir \cite{DBLP:conf/focs/LapidotS91},
Feige-Lovasz \cite{Feige:1992:TOP:129712.129783} and Dwork, Feige, Kilian, Naor, and Safra, \cite{DBLP:conf/crypto/DworkFKNS92} to achieve a ``2-prover 1-round 0-knowledge'' proof for $\mathbf{NEXP}$.
As far as we can tell, this is the only paper in the ZKMIP literature that appears to avoid the multi-round problems and the non-local contamination that we discuss.
However, note that the analysis of \cite{DBLP:conf/stoc/FeigeK94} is partly based of that of \cite{DBLP:conf/focs/LapidotS91}, and the journal version of
Feige-Kilian \cite{DBLP:journals/siamcomp/FeigeK00} does not contain their prior claim of zero-knowledge either.
All other ZKMIPs for $\mathbf{NEXP}$ in the literature are multi-round, and thus our analysis applies to them.

Similar issues are possible using more recent results such as Ito-Vidick's proof \cite{IV12} that $\mathbf{NEXP}\subseteq \mathbf{MIP}^{*}$ and
Kalai, Raz and Rothblum's proof \cite{KRR14} that $\mathbf{MIP}^{ns}\! = \mathbf{EXP}$.
The reason why these multi-round constructions may maintain their soundness despite the potential non-locality contamination (via the verifier) is the {\em non-adaptive} nature of their verifiers.
Non-adaptive verifiers cannot take advantage of information acquired in recent rounds to construct new questions to the provers: all their questions are pre-established before the interaction
with the provers start. This is a special simpler case of local verifiers. Nowhere in this large literature can one find a single statement observing the non-adaptiveness of the verifiers and
its importance to guarantee soundness of those MIPs.
Moreover, their multi-round structure requires that any straightforward extensions to $\mathbf{ZKMIP}^{*}$ or $\mathbf{ZKMIP}^{ns}$ via commitment schemes be analyzed very carefully
and the locality of the resulting verifiers be re-established. This is part of the reasons why the ZK version did not follow easily.
Recently, Chiesa, Forbes, Gur, and Spooner \cite{DBLP:journals/eccc/ChiesaFGS18} discovered a proof that $\mathbf{NEXP}\subseteq \mathbf{ZKMIP}^{*}$.
Their construction is based on refinements of Ito-Vidick's proof and along the lines of Feige-Kilian, building on algebraic structures to bypass the need of commitment schemes.
Unfortunately, this work is so complicated that we are unable to assess whether their verifier is actually non-adaptive.
And of course, this is not mentioned or proven anywhere nor available from the authors...

Bellare, Feige, and Kilian \cite{DBLP:conf/istcs/BellareFK95} considered a multi-verifier model similar to ours in order to analyze the role of randomness in multi-prover proofs.
This is completely unrelated to our goal of analyzing verifier non-local contamination.
Finally, the notion of relativistic commitment schemes put forward by Kilian \cite{kilian1990strong} and Kent \cite{PhysRevLett.83.1447} leads to several results \cite{PhysRevLett.115.030502,DBLP:journals/corr/AdlamK15,CL17} where a similar multi-verifier model is necessary in order to assess spatial separation of the provers.
The new (Non-local) Zero-Knowledge definition is 100\% fresh from this work. No prior work exists at all.

\section{The Standard MIP Model}\label{SEC:Standard}

Multi-prover interactive proofs were introduced in \cite{BGKW88}. The intuition for their model was that of a detective interrogating two suspects held in different rooms. This was formalized as follows:

\begin{definition}

Let $P_1, \ldots, P_k$ be computationally unbounded Turing machines and let $V$ be a probabilistic polynomial-time Turing machine. All machines have a read-only input tape, a read-only auxiliary-input tape, a private work tape and a random tape. The $P_i$'s share a joint, infinitely long, read-only random tape. Each $P_i$ has a write-only communication tape to $V$, and vice-versa. We call $(P_1, \ldots, P_k, V)$ a \emph{k-prover interactive protocol (k-prover IP)}.

\end{definition}

This model is essentially equivalent to that of Bell \cite{BELL64} who introduced his famous Bell's inequality to distinguish {\em local} parties from {\em entangled} parties.

Zero-knowledge MIPs were also defined in \cite{BGKW88}:

\begin{definition}\label{DEF:standardZK}

Let $(P_1, \ldots, P_k, V)$ be a k-prover IP for a language $L$.Let $\mathbf{view}(P_1, \ldots, P_k, V, x)$ denote the verifier's incoming and outgoing messages with the provers, including his coin tosses. We say that $(P_1, \ldots, P_k, V)$ is \emph{perfect zero-knowledge} {for} $L$ if there exists an expected polynomial-time machine $M$ such that for all $V'$, $\mathbf{view}(P_1, \ldots, P_k, V', x)$ and $M(x)$ are identically distributed.

\end{definition}

Let us call the above two definitions the \emph{standard MIP model}. There have also been augmentations of the model by giving the {provers} various non-local resources, such as entanglement \cite{IV12}, or arbitrary no-signaling power \cite{KRR14}. 

The first work to point out the aforementioned blind spot in the standard MIP model, although it was not worded explicitly, was \cite{CSST11}. In order to understand their point, we need to understand the following two-prover protocol.

\begin{quote}
\rule{\linewidth}{1pt}
\begin{protocol}\label{PROT:CSST_BC} ( BGKW-type commitment for bit $b$ ) \end{protocol}

$P_1$ and $P_2$ pre-share a random $n$-bit string $w$.

\begin{enumerate}

\item $V$ sends a random $n$-bit strings $r$ to $P_2$.

\item $P_2$ replies with $x \leftarrow b \times r \oplus w$.

\item $P_1$ announces to $V$ a string $w'$.

\item $V$ accepts iff $(w' \oplus x) \in \{0,r\}$.

\end{enumerate}

\rule{\linewidth}{1pt}
\end{quote}

This is a two-prover commitment protocol. Steps 1 and 2 commit, while steps 3 and 4 unveil. An intuitive proof of its binding condition is that, since the provers cannot signal, and they both need to know $r$ in order to unveil the commitment in the way they want, therefore they cannot cheat. This intuition is incomplete, as was pointed out in \cite{CSST11}, because breaking the binding condition \emph{does not require signaling}. The following protocol, known as a {\bf PR}-box, can be used to break binding without signaling.

\begin{figure}[h!]
\centering
  \mbox{\Qcircuit @C=1em @R=.7em {
      \lstick{c}  \ar[r] & \multigate{1}{{\bf PR}} & \rstick{r} \ar[l] \\
      \lstick{w' := c \times r \oplus x}   & \ghost{{\bf PR}} \ar[r]\ar[l] & \rstick{x \text{ (uniform)}}
    }}
\caption{a {\bf PR}-box}
  \label{nlbox}
\end{figure}
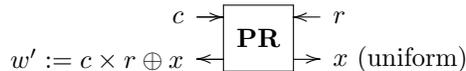

By having $P_1,P_{2}$ obtain $w',x$ via the PR-box, $P_1$ can unveil the commitment the way it wishes, $c$. This fact will become extremely important in Sections \ref{SEC:NEXP} and \ref{SEC:NEWZK}.
                 
The punchline of \cite{CSST11} is that \emph{the verifier itself can act as a PR-box for the provers without violating their no-signaling assumption}. Consider the following:


\begin{enumerate}

\item Any security proof of protocol \ref{PROT:CSST_BC} must show that it does not contain a PR-box as a subroutine.

\item More generally, any security proof of a protocol must show that no subroutine within itself can be commandeered by the provers to achieve a non-local functionally (like the PR-box).

\item Composition of protocols, for instance between the committing and the opening of commitments, must be done in such a way that provably does not create a non-local box.

\end{enumerate}

The solution proposed in \cite{CSST11} was that of \emph{verifier isolation}. Informally, this means that any message an ``isolating'' verifier sends to a set $S$ of provers must be computed solely from messages that are received from $S$. The end result is that an isolating verifier can never accidentally implement a PR-box and, in general, it will always enforce the locality of the provers. In a sense, we can think of an isolating verifier as ``local''. Our new model will make this more precise and more general.


Furthermore, existing zero-knowledge MIPs such as \cite{KILIAN89} \emph{require} that the verifier courier an authenticated message between the provers in order to obtain soundness while ensuring zero-knowledge. The gist of it goes like this:

\begin{enumerate}

\item $V$ asks $P_1$ some questions.

\item $V$ wants to check one of $P_1$'s answers with $P_2$ for consistency.

\item In order for zero-knowledge to hold, $V$ \emph{must} ask $P_2$ a question it has already asked $P_1$.

\item $P_1$ authenticates a question with a key that was committed at the beginning of the protocol and sends it to $V$.

\item $V$ sends the question and the authentication to $P_2$, who proceeds only if authentication succeeds.

\end{enumerate}

Steps 4 and 5 consists of $V$ sending a message from $P_1$ to $P_2$. Proofs that this act does not contaminate non-locally (such as simulating a PR-box) is not found in any existing MIP.
This needs to be proven, and the proof contained in \cite{KILIAN89} does not address this issue. Moreover, the zero-knowledge protocol of \cite{KILIAN89} allows $P_1$ to send an arbitrary message to $P_2$ (via the authentication key). Therefore, one cannot compose such a protocol in a nested fashion (as a subroutine call) since the inner instance would violate the no-communication assumption of the outer instance. For more details on the problems of the standard $\mathbf{MIP}$ model, see \cite{Crepeau2017}.

Existing simulators for zero-knowledge protocols such as those found in \cite{KILIAN89} needs to know how to break commitments in order to simulate. The simulator accomplishes this by acting as both provers, thereby receiving the secret string $r$ which was meant for one prover only. This standard model of zero-knowledge gives the simulator \emph{unnecessary power}, in a sense. We will discuss this further in section \ref{SEC:NEWZK}.

\section{Locality-Explicit MIP}\label{SEC:CCMIP}

The standard MIP model allows the verifier to non-locally contaminate the provers. We neutralize this problem by defining a model with multiple verifiers, each of which talks to a single prover; in turn, each prover talks to a single verifier. There are no communication tapes between the verifiers, nor are there between provers. There is a special verifier $V_{0}$ which \emph{only reads} the outputs of the other verifiers; this is the verifier that will decide to accept or reject membership to $L$. We call this model ``locality-explicit'' since the provers and verifiers are explicitly local, and if any non-local resources (such as entanglement) are available to them, then it is explicitly specified via a supplementary entity named $\widehat{P}$ for the provers and $\widehat{V}$ for the verifiers.


This model is a \emph{generalization} of the standard model because the special setting where $\widehat{P}$ is empty and $\widehat{V}$ signals for the verifiers corresponds to the standard MIP model.


\begin{definition}

An \emph{interactive Turning machine} (ITM) is a Turing machine augmented with the following tapes:

\begin{itemize}

\item $k_1$ read-only incoming communication tapes.

\item $k_2$ write-only outgoing communication tapes.

\item Private work, auxiliary-input, and random tapes.

\end{itemize}

An ITM $A$ can signal to an ITM $B$ if $A$'s write-only outgoing tape is $B$'s read-only incoming tape.

\end{definition}

\begin{definition}
Let $(\widehat{P}, P_1, \ldots, P_k, \widehat{V}, V_0, V_1, \ldots, V_k)$ be a tuple of ITMs, where the $P\!$'s are computationally all-powerful and the $V\!$'s are polynomial-time. 
For each $i$, there are two-way communication tapes between $V_i$ and $P_i$, and that for all $j$, there is a two-way communication tape between $\widehat{V}$ and $V_j$ and also between $\widehat{P}$ and $P_j$. In addition, for each $\ell$, there is a read-only tape going from $V_\ell$ to $V_0$ (where $V_0$ reads). Then, this is said to be a \emph{locality-explicit multi-prover interactive proof}.

We call $\widehat{P}$ and $\widehat{V}$ \emph{correlators} and say that the provers and verifiers are $\widehat{P}$-\emph{local} and $\widehat{V}$-\emph{local} respectively. We define the class of all MIPs with such correlators $\mathbf{MIP}^{\widehat{P}}_{\widehat{V}}$.

\end{definition}

It is perhaps easier to understand our definition with the help of figure 2.

\begin{figure}[h]
\label{FIG:CCMIP}
\centering
\includegraphics[angle=0, width=0.75\textwidth]{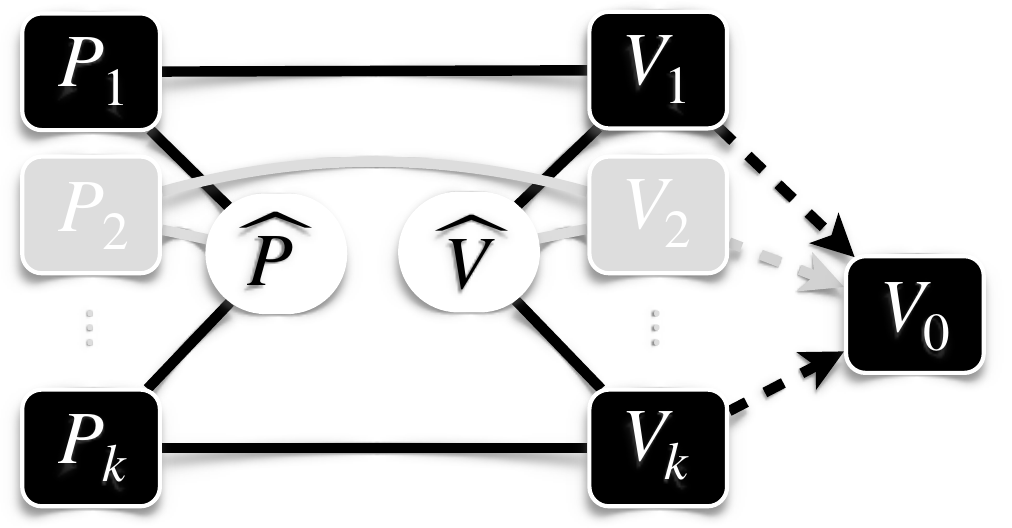}
\caption{Locality-Explicit MIP}
\end{figure}

The solid lines represents two-way communication and the dashed arrows represents one-way communication, with the arrow indicating the direction of information flow. 

We can define that an LE-MIP accepts a language $L$ if the usual soundness and completeness conditions hold:

\begin{definition}

An LE-MIP $(\widehat{V}, V_0, V_1, \ldots, V_k, \widehat{P}, P_1, \ldots, P_k)$ accepts a language $L$ if and only if

\begin{itemize}

\item (completeness) $\forall x \in L, \mathbf{Pr}[V_0(x, t_1, \ldots, t_k) = \emph{\textsf{accept}}] > 2/3$,

\item (soundness) $\forall x \notin L, \forall P_1^{\prime}, \ldots, P_k^{\prime}, \mathbf{Pr}[V_0(x, t_1, \ldots, t_k) = \emph{\textsf{accept}}] < 1/3$,

\end{itemize}

where $t_i$ is the read-only tape from $V_i$ to $V_0$ at the end of the interaction of $V_i$ with $P_i$ (or $P_i^{\prime}$) on input $x$.

\end{definition}

Note that we do not quantify over $\widehat{P}$ (nor $\widehat{V}$), as we want to use them not as (possibly malicious) participants to the protocol, but as a description of non-local resources available to the provers and verifiers.


\begin{definition}\label{DEFMIP}

An LE-MIP is \emph{local} if $\widehat{V} = \widehat{P} = \emptyset$ and all of the provers' (resp.~verifiers') random tapes are initialized with the same uniformly random string $R$ (resp.~verifiers with another, independent uniformly random string $S$)\footnote{By $\emptyset$ we mean the empty correlator that provides everyone with nothing at all as output.}.

\end{definition}

MIPs in the standard model (with local provers) are equivalent to LE-MIPs where $\widehat{P} = \emptyset$ and $\widehat{V}$ acts as a bulletin board. That is, a single verifier communicating with multiple provers is equivalent to multiple verifiers communicating with provers and each other.


In standard MIPs, it is possible that the honest (single) verifier bridges the provers non-locally. If a protocol does not desire this -- and most existing MIPs do not -- it must be proven. With local LE-MIPs, the special verifier $V_0$ decides to accept or reject. This verifier cannot communicate with anyone else, avoiding the aforementioned problem of contamination.

%
%
%
%
%

\subsection{Zero-Knowledge LE-MIPs}\label{SEC:NEWZK}


%

Zero-knowledge is defined by simulations, the fundamental idea that if a transcript can be produced by an entity (simulator) with no more power than one (verifier) interrogating all-powerful provers, then no knowledge is gained.

The simulator of single-prover IP and standard MIP are equal to the verifier in computational power, but they do have ``advantages'' which allow them to fake transcripts. For single-prover IPs, the simulator is allowed to rewind computation; for standard MIPs, the simulator is given a (commitment-breaking) secret. Those advantages are, of course, independent of knowledge.

LE-MIPs naturally induces a new advantage for the simulator: non-local correlations. This is a very powerful advantage. Using the correct non-local correlations, simulators do not need to rewind, do not need to pretend to be multiple (isolated) provers, and do not need to know any commitment-breaking secrets. In short, they do not need to signal. Multiple, no-signaling simulators can even produce transcripts in ``real-time'' (example will follow) if the proper correlations are used.


%
%

\begin{definition}

Let $\mathcal{M} = (\widehat{M}, M_1, \ldots, M_k)$ be a tuple of polynomial-time ITMs. Each machine has a random tape, and every random tape is initialized with the same random bits. For $1 \leq i \leq k$, there is a two-way communication tape between $\widehat{M}$ and $M_i$. There are no communication tapes between any of the $M_i\!$'s. Then this is called a tuple of \emph{locality-explicit simulators} and $\widehat{M}$ is the \emph{locality class} of $\mathcal{M}$, which will be abbreviated $\widehat{M}$\emph{-local}.

\end{definition}
%
%
%
%
%

\begin{definition}[White-box version]

Let $\mathcal{PV}=(\widehat{P}, P_1, \ldots, P_k, \widehat{V}, V_0, V_1, \ldots, V_k)$ be an LE-MIP for language $L$. If there exists a correlator $\widehat{S}$ such that for all verifiers $(V^\prime_0, V^\prime_1, \ldots, V^\prime_k)$,
there exists $(S_1, \ldots, S_k)$ for all correlator $\widehat{V}'$, such that for all $x\in L$ the transcripts of conversations $$(\widehat{P}, P_1, \ldots, P_k, \widehat{V}', V^\prime_0, V^\prime_1, \ldots, V^\prime_k)(x)$$
and those generated by $$(\widehat{S} \cup \widehat{V}', V^\prime_0, S_1, \ldots, S_k)(x)$$ are identically distributed, where $(\widehat{S}, S_1, \ldots, S_k)$ is a tuple of locality-explicit simulators, then we say that $\mathcal{PV}$
is a \emph{$\widehat{S}$-local perfect zero-knowledge LE-MIP} for $L$.

We will denote the set of all ZK LE-MIPs where the provers, verifiers, and simulators are $\widehat{P}$-local, $\widehat{V}$-local, and $\widehat{S}$-local by $$\mathbf{ZK}^{\widehat{S}}\mathbf{MIP}^{\widehat{P}}_{\widehat{V}}.$$

Let $\mathbb{S, P, V}$ be sets of correlators. We will denote, by convention, $$\mathbf{ZK}^{\mathbb{S}} \mathbf{MIP}^{\mathbb{P}}_{\mathbb{V}}$$ as the set of all ZK LE-MIPs where each correlator comes from each of the respective sets.

\end{definition}

\begin{definition}[Black-box version]

Let $\mathcal{PV}=(\widehat{P}, P_1, \ldots, P_k, \widehat{V}, V_0, V_1, \ldots, V_k)$ be an LE-MIP for language $L$. If there exists a tuple of locality-explicit simulators $(\widehat{S}, S_1, \ldots, S_k)$,
such that for all verifiers $(\widehat{V}',V^\prime_0, V^\prime_1, \ldots, V^\prime_k)$, such that for all $x\in L$ the transcripts of conversations $$(\widehat{P}, P_1, \ldots, P_k, \widehat{V}', V^\prime_0, V^\prime_1, \ldots, V^\prime_k)(x)$$
and those generated by $$(\widehat{S}, V^\prime_0, S_1(V^\prime_{1}), \ldots, S_k(V^\prime_{k}))(x)$$ (where the $V^\prime_{i}$ still have access to $\widehat{V}'$) are identically distributed, then we say that $\mathcal{PV}$ 
is a \emph{$\widehat{S}$-local perfect (black-box) zero-knowledge LE-MIP} for $L$.

We will denote the set of all BBZK LE-MIPs where the provers, verifiers, and simulators are $\widehat{P}$-local, $\widehat{V}$-local, and $\widehat{S}$-local by $$\mathbf{ZK}^{\widehat{S}}_{\mathbf{BB}}\mathbf{MIP}^{\widehat{P}}_{\widehat{V}}.$$

Let $\mathbb{S, P, V}$ be sets of correlators. We will denote, by convention, $$\mathbf{ZK}^{\mathbb{S}}_{\mathbf{BB}} \mathbf{MIP}^{\mathbb{P}}_{\mathbb{V}}$$ as the set of all BBZK LE-MIPs where each correlator comes from each of the respective sets.

\end{definition}

Our motivations for the above definitions are twofold.

First, a simulator (or simulators) should not have more power than necessary. If two \emph{local} simulators can output for two \emph{local} verifiers, then it is not necessary to have a single simulator (equivalent to two \emph{signaling} simulators) do the job. Allowing simulators to signal (equivalently, having a single simulator) in the multi-prover setting is analogous to allowing unbounded running-time simulation in single-prover zero-knowledge. In general, finding the minimal $\widehat{S}$ that will allow simulation establishes how little extra is needed to obtain the zero-knowledge property.

Second, the non-locality of simulators is a characterization of the resilience of zero-knowledge. A protocol with local simulators which can withstand arbitrary (malicious) verifiers is more resilient than one in which signaling simulators are needed.

This may be of practical interest, if transcripts are timestamped. For example, under the relativistic assumption that one may not signal faster-than-light, one may be able to distinguish two spatially separated simulators from two spatially separated verifiers, if the simulators need to signal (transmit a commitment-breaking secret) in order to generate a transcript. On the other hand, if two entangled simulators are sufficient to produce the transcript, then they are indistinguishable from real verifiers and provers. Our protocol \ref{IsolatingZKMIP} can be modified as to let entangled simulators do their work, without needing PR-boxes or signaling. Details in section \ref{SEC:NEXP}.


%
%
%
%
%
%
%
%
%

\subsection{The Power of LE-MIPs}\label{SEC:power}

Local LE-MIPs form a subclass of standard MIPs. They are, by design, more restricted in what you can make the verifier do. An immediate question is whether this is \emph{too} restrictive. Perhaps, in all interesting cases, it is necessary for a single verifier to go back-and-fourth between provers, using previous discussions to generate new questions.

The answer is that, of all the literature we have surveyed, almost all protocols can be re-written in a local-verifier manner without any loss of functionality. We explicitly demonstrate this for the multi-prover protocol for oracle-3-SAT in \cite{BFL90}. The protocol details can be found in the appendix. For the purpose of our discussion, we only need to look at the general form of the protocol:

\begin{quote}
\rule{\linewidth}{1pt}
\begin{protocol}\label{BFL_classic}( BFL Classic, Single-Verifier )\end{protocol}

\begin{enumerate}

\item $V$ asks $P_1$ some questions non-adaptively.

\item $V$ chooses a question $Q$ from the pool of questions which were asked to $P_1$.

\item $V$ asks $Q$ to $P_2$.

\item $V$ accepts if the interaction with $P_1$ was successful, and the answer from $P_2$ is consistent with those of $P_1$.

\end{enumerate}

\rule{\linewidth}{1pt}
\end{quote}

The crucial observation is that $V$ does not \emph{adaptively} ask questions to $P_1$. Therefore, the questions asked on that entire side of the conversation can be selected in advance, and thus they can be shared in advance with a second verifier. We can therefore naturally rewrite the BFL classic protocol as a local LE-MIP in the following way. The reader can check the details in the appendix, and in section 3 of \cite{BFL90}.

\begin{quote}
\rule{\linewidth}{1pt}
\begin{protocol} ( BFL as an LE-MIP ) \end{protocol}

\begin{enumerate}

\item $V_1$ prepares the questions which it will ask $P_1$.

\item $V_1$ chooses a question $Q$ from the above list and shares it with $V_2$.

\item LE-MIP begins. All parties are local as per definitions.

\item $V_1$ asks the questions to $P_1$.

\item $V_2$ asks $Q$ to $P_2$.

\item $V_0$, reading the responses, decides to accept or reject, based on the same criteria as in protocol \ref{BFL_classic}.

\end{enumerate}

\rule{\linewidth}{1pt}
\end{quote}

The BFL protocol is for oracle-3-SAT, which is $\mathbf{NEXP}$-complete. Rewritten as a local LE-MIP, it circumvents all non-locality issues we have mentioned. Thus, we can conclusively say that ``$\mathbf{MIP}^{\emptyset}_{\emptyset} = \mathbf{MIP} = \mathbf{NEXP}$''; no transformation to single-round MIP necessary, and no need to invoke the general theory of PCPs.

\section{$\mathbf{ZK}^{\mathbf{PR}}\mathbf{MIP}^{\emptyset}_{\emptyset} = \mathbf{NEXP}$}\label{SEC:NEXP}

The question which follows naturally is whether there exists a \emph{zero-knowledge}, local LE-MIP for $\mathbf{NEXP}$. The existing technique for achieving zero-knowledge in MIP \cite{BGKW88,KILIAN89} requires the (single) verifier to courier an authenticated message between provers. This is not possible with local-verifier LE-MIPs. We show that there is a way around that constraint.

By adapting the protocol from \cite{BFL90}, we will exhibit a protocol with the following properties:

\begin{enumerate}

\item The provers and verifiers are local: $\widehat{V} = \widehat{P} = \emptyset$.

\item The simulators need only access to instances of PR-boxes to work. That is, $\widehat{M}$ simply computes indexed instances of $\mathbf{PR}$-boxes. We will abbreviate this as ``$\mathbf{PR}$-local.''

\end{enumerate}

We may succinctly summarize the above as $\mathbf{ZK}^{\mathbf{PR}}\mathbf{MIP}^{\emptyset}_{\emptyset} = \mathbf{NEXP}$, where $\mathbf{PR}$ denotes a correlator which simply computes $\mathbf{PR}$-boxes for the simulators.



The generic way of turning an interactive proof into a zero-knowledge one is by running it in committed form \cite{BGKW88,KILIAN89}. With this technique, provers commit their answers instead of directly responding, and use cryptographic techniques to convince the verifier that the answers are correct.

As shown in section \ref{SEC:power}, the BFL protocol can be turned into a local LE-MIP. If we try to turn it into a zero-knowledge LE-MIP by having the provers commit their answers (for example using protocol \ref{PROT:CSST_BC} as commitment), we run into a problem. In order to achieve zero-knowledge, the provers \emph{must} ensure that the question $P_2$ receives from $V_{2}$ is one of the questions which $V_{1}$ has asked $P_1$. On the other hand, since the provers and verifiers are local, the provers cannot communicate, nor can they ask the verifiers to courier authenticated messages between them.


Our solution essentially asks the provers to (strongly-universal-2) hash the selected committed answer with a key that is based on the verifier's question. We force $V_2$ to behave honestly (to ask a question that $V_1$ has asked) by making bad questions meaningless. If the verifiers ask the provers the same question, they will receive the same hash of the same answer. Otherwise, they will receive two unrelated random hash values. 




We need the $\mathbf{PR}$ commitment (protocol \ref{MIP_CHSHBC}), which is secure in the local setting as previously proved in \cite{PhysRevLett.83.1447,CSST11,PhysRevLett.115.030502}.

\subsection{The Protocols}
\label{MIP_local}

The following is a $\mathbf{PR}$-type commitment that is perfectly concealing and statistically binding.
In general, we use the commitment-box notation ``\fbox{$b$}'' as the name of a commitment to bit $b$ in the next two protocols.

\begin{quote}
\rule{\linewidth}{1pt}
\begin{protocol}\label{MIP_CHSHBC}\label{MIP_CHSH} A statistically binding, perfectly concealing commitment protocol to bit $b$.\end{protocol}

All parties agree on a security parameter $1^{k}$.\\ 
$P_1$ and $P_2$ partition their private random tape into two $k$-bit strings $w_{1},w_{2}$.\\

{\bf Pre-computation phase:}

\begin{itemize}

\item $V_{1}$ samples two $k$-bit strings $z_{1},z_{2}$ independently and uniformly, and provides them to $V_{2}$.

\item $V_{1}$ sends $z_{1}$ to $P_1$ and $V_{2}$ sends $z_{2}$ to $P_2$.

\end{itemize}

{\bf Commit phase:}

\begin{itemize}

\item $P_1$  commits $b$ to $V_{1}$ as \fbox{$b$} $ = (b \times z_{1})  \oplus  w_{1}$, where $b \times z_{1}$ is a multiplication in $\mathbb{F}_{2^{n}}$.

\item $P_2$ sends $V_{2}$: $d = (w_{1} \times z_{2}) \oplus w_{2}$.

\end{itemize}

{\bf Unveiling phase:}

\begin{itemize}

\item $P_1$ sends $w_{1},w_{2}$ to $V_{1}$.

\item $V_{1}$ computes $b = 1$ if $\text{\fbox{$b$}} \oplus  w_{1} = z_{1}$, or $b =0$ if $\text{\fbox{$b$}} =  w_{1}$. 

\item $V_{0}$ {\bf rejects} if $\text{\fbox{$b$}} \oplus w_{1}$ is anything but $z_{1}$ or $0$, or if $d \oplus w_{2} \neq w_{1} \times z_{2}$
and {\bf accepts} $b$ otherwise.

\end{itemize}

\rule{\linewidth}{1pt}
\end{quote}

Below is the zero-knowledge, local LE-MIP for oracle-3-SAT (Protocol \ref{IsolatingZKMIP}). The basis of protocol \ref{IsolatingZKMIP} is the localized BFL protocol we presented in section \ref{SEC:power} (details in the appendix). 
A note on notation: for a circuit $f$, we will denote $f\!\left( \text{\fbox{$x$}} \right)$ as the gate-by-gate committed circuit evaluated with x as the input.
We also use statements such as ``$P_1$ proves to $V_1$ that \text{\fbox{$\Omega_1$}} was computed correctly''. The reader is expected familiarity with zero-knowledge computations on committed circuits as put forward by \cite{BrassardCa86,BrassardCb86,ImpagliazzoY87,KILIAN89}.

\begin{quote}
\rule{\linewidth}{1pt}
\begin{protocol}\label{IsolatingZKMIP} A local zero-knowledge LE-MIP for oracle-3-SAT\end{protocol}

Let $x=(B,r,s)$, an instance of oracle-3-SAT, be the common input, let $k = \left| x \right| = r+3s+3$, and let $\Lambda$ be the verifier's program in protocol \ref{PROT:BFL_appendix} (see appendix).

\begin{enumerate}

\item {\bf Pre-computation:}

	\begin{enumerate}

	\item $V_{1}$ samples two $k$-bit strings $z_{1},z_{2}$ independently and uniformly, and provides them to $V_{2}$.

	\item $V_{1}$ selects $k+3$ random bit strings $R_1,...,R_{k+3}$ (size specified implicitly by $\Lambda$) and evaluates the circuit of $\Lambda$ using the $R_i$ as randomness, resulting in questions $Q_1,...,Q_{k+3}$, and provides them to $V_{2}$

	\item $V_{1}$ randomly chooses $i$, $1 \leq i \leq k+3$, the index of an oracle query that will be made to both $P_1$ and $P_{2}$. $V_{1}$ provides $i$ to $V_{2}$.

	\item $V_{1}$ sends $z_{1}$ to $P_1$ and $V_{2}$ sends $z_{2}$ to $P_2$ for future commitments.
	
	\item All parties agree on a family of strongly-universal-2 hash functions $\{H_i\}$ indexed by $k$-bit keys.
	
	\item $P_1$ and $P_2$ agree on a $k$-bit key $\gamma$, an index to the above family.
	
	\item $P_1$ commits \fbox{$\gamma$} to $V_{1}$.

	\end{enumerate}

\item {\bf Sumcheck with oracle:}

	\begin{itemize}

	\item Let $f$ be the arithmetization obtained in protocol \ref{PROT:BFL_sumcheck}, 
let $z$ be a string from $I^r$ and $Q_{k+1}, Q_{k+2}, Q_{k+3}$ be strings of $I^s$ as generated in protocol \ref{PROT:BFL_appendix}. 
$V_{1}$ and $P_1$ execute protocol \ref{PROT:BFL_sumcheck} in committed form.
At the end of this phase, $P_1$ shows that the committed final value is equal to
$$f\!\left( z, Q_{k+1}, Q_{k+2}, Q_{k+3}, \text{\fbox{$A(Q_{k+1})$}}, \text{\fbox{$A(Q_{k+2})$}}, \text{\fbox{$A(Q_{k+3})$}} \right), $$
an evaluation in committed form of $f$ using the committed values that were used during the protocol's loop. If this fails, $V_{1}$ instructs $V_{0}$ to reject.

	\end{itemize}

\item {\bf Multilinearity test:}

\begin{enumerate}

	\item For $1 \leq i \leq k$:

	\begin{enumerate}
		\item $V_1$  sends $Q_{i}$ to $P_{1}$,
		\item $P_1$ commits his answer as \fbox{$A(Q_i)$}.
	\end{enumerate}
	
	\item $P_1$ and $V_{1}$ evaluate a circuit description of $\Lambda$ in committed form with inputs $\text{\fbox{$A(Q_1)$}}, \ldots, \text{\fbox{$A(Q_k)$}}$ to verify proper linearity among them.
$P_1$ unveils the circuit's committed output. If it rejects, $V_{1}$ instructs $V_{0}$ to reject.

	\end{enumerate}

\item {\bf Consistency test:}

	\begin{enumerate}

	\item $V_1$ sends $i$ to $P_1$.
	
	\item $P_1$ computes $\text{\fbox{$\Omega_1$}} = \text{\fbox{$A(Q_i)$}} \oplus H_{\text{\fbox{$\gamma$}}}\!\left( Q_{i}  \right)$ and sends \text{\fbox{$\Omega_1$}} to $V_1$.
	
	\item $P_1$ proves to $V_1$ that \text{\fbox{$\Omega_1$}} was computed correctly, from the existing commitments.
	
	\item $P_1$ unveils \text{\fbox{$\Omega_1$}} for $V_1$, who gets $\Omega_1$.

	\item $V_{2}$ sends $Q_{i}$ to $P_2$ (recall that this was pre-agreed in step 1.(c))

	\item $P_2$ responds to $V_2$ with $\Omega_2 = A(Q_i) \oplus H_{\gamma}\!\left( Q_{i}  \right)$. 

	\item $V_{0}$ accepts if and only if all of the following conditions are met:
	
		\begin{itemize}
	
		\item $\Omega_1 = \Omega_2$
		\item All commitments which have been unveiled are valid.
		\item $V_{1}$ did not reject in the two previous 	cases
		
		\end{itemize}

	\end{enumerate}

\end{enumerate}

\rule{\linewidth}{1pt}
\end{quote}

\subsection{Proofs of Security}\label{PofS}

\subsubsection{Locality}~\\


Since the protocol is written as an LE-MIP in which $\widehat{P} = \widehat{V} = \emptyset$, the protocol is local by definition \ref{DEFMIP}.

\subsubsection{Completeness}~\\

Completeness follows from the completeness of the underlying protocol \cite{BFL90}, and the fact that the commitment protocol (protocol \ref{MIP_CHSH}) is well-defined for honest provers (who will never send a commitment that they cannot unveil).

\subsubsection{Soundness}~\\


Without loss of generality, we may assume that the soundness error in the BFL protocol to be $1/3$, through sequential amplification. The probability that our commitment scheme (protocol \ref{MIP_CHSH}) fails binding is exponentially small in $k$.
Local probabilistic provers are equivalent to local deterministic provers. This is because the success probability $\alpha$ of randomized provers of breaking soundness is an average over the randomized provers' random tapes. Each instance of a random tape represents a deterministic strategy. Therefore there is a deterministic strategy which succeeds with probability at least $\alpha$, and hence we only need to consider local deterministic provers.

Since $P_1$ is deterministic, we may unambiguously consider what happens if we were to ``rewind'' the prover machine. Suppose that at some point $P_1$ unveils a particular commitment $c$ to $0$. We rewind $P_1$ and let $V_1$ make different choices before that point.
Suppose that, with these alternate choices, $P_1$ then unveils $c$ to $1$ (an attempt to break binding). Because of locality, $P_1$'s behavior is independent of what $P_2$ receives (namely $z_2$). Therefore, there is only \emph{one} such $z_2$ which $V_0$ will ultimately accept as a valid unveiling of $c$ in both ways (recall that our commitment is statistically binding).

Therefore, in the worst case, for every commitment there exists a sequence of interactions between $V_1$ and $P_1$ such that $P_1$ will attempt to break the binding of that commitment. Each such commitment-breaking corresponds to at most one string $z_2$ that will actually work.

Let us denote the set of such binding-breaking strings by $B$. If $z_2 \notin B$, then the provers \emph{will not break binding}, and the soundness error is reduced to that of the underlying protocol (at most $1/3$). On the other hand, since $\left| B \right| < \mathbf{poly}(k)$, the probability that $z_2 \in B$ is at most $ \mathbf{poly}(k)/2^k$.

Therefore, the soundness error of our protocol is at most $$Pr[ z_2 \notin B\text { and underlying protocol accepts}] + Pr[z_2 \in B] \leq \frac{1}{3} + \frac{\mathbf{poly}(k)}{2^k}.$$

\subsubsection{Zero-Knowledge}

The simulation will be divided in two parts. In the first part, the simulator produces a transcript of the \emph{pre-computation}, \emph{multilinearity test} and \emph{sumcheck with oracle} parts, which involves only interactions with $V_1$. In the second part, the simulator will fake a valid \emph{consistency test}.

\begin{quote}
\rule{\linewidth}{1pt}
\begin{protocol}\label{CCSIM1} ( Perfectly Indistinguishable, $\mathbf{PR}$-Local Simulator for Protocol \ref{IsolatingZKMIP}, Part 1) \end{protocol}

The setup:

\begin{itemize}

\item Let $( \widehat{S}, S_1, S_2)$ be a set of locality-explicit simulators.

\item $S_1$ and $S_2$ can send $\widehat{S}$ an index along with a bit.

\item $\widehat{S}$ completes the indexed $\mathbf{PR}$ box (protocol \ref{PROT:CSST_BC}) for both simulators.

\end{itemize}

The simulation strategy:

\begin{enumerate}

\item The simulators agree on unique indices for every commitment used in the protocol.

\item $S_1$ interacts with $V_1$ the way $P_1$ would. Whenever $P_1$ should commit, $S_1$ commits to random bits, just like the single-simulator from Sec.~\ref{SEC:NEXP}.

\item For each commitment, $V_2$ sends $S_2$ a string $s$. $S_2$ sends to $\widehat{S}$ the index of the commitment and $s$.

\item $\widehat{S}$ runs the $\mathbf{PR}$ box (protocol \ref{PROT:CSST_BC}) and replies with $V_2\!$'s half of the output.

\item Whenever $S_1$ needs to unveil a commitment, it can be unveiled in the way $S_1$ desires by sending the corresponding index and bit to $\widehat{S}$.

\item $\widehat{S}$ completes the corresponding $\mathbf{PR}$ box which outputs $t$. $\widehat{S}$ sends $t$ to $S_1$.

\item $S_1$ sends $t$ to $V_1$.

\end{enumerate}

\rule{\linewidth}{1pt}
\end{quote}

The second part (the consistency test) can be done by having the simulators ignore the question.

\begin{quote}
\rule{\linewidth}{1pt}
\begin{protocol}\label{CCSIM2} ( Perfectly Indistinguishable, $\mathbf{PR}$-Local Simulator for Protocol \ref{IsolatingZKMIP}, Part 2) \end{protocol}

\begin{enumerate}

\item $V_1$ sends $i$ to $S_1$.

\item $S_1$ computes $\text{\fbox{$\Omega_1$}} = H_{\text{\fbox{$\gamma$}}}\!\left( Q_{i}  \right)$.

\item Using $\widehat{S}$ to break binding, $S_1$ convinces $V_1$ that $\text{\fbox{$\Omega_1$}}$ is actually $\text{\fbox{$A(Q_i)$}} \oplus H_{\text{\fbox{$\gamma$}}}\!\left( Q_{i}  \right)$.

\item $S_1$ unveils $\text{\fbox{$\Omega_1$}}$ for $V_1$, who gets $\Omega_1 = H_{\text{{$\gamma$}}}\!\left( Q_{i}  \right)$.

\item $V_2$ sends $Q_i'$ to $S_2$.

\item $S_2$ responds with $\Omega_2 = H_{\text{{$\gamma$}}}\!\left( Q_{i}'  \right)$.

\end{enumerate}

\rule{\linewidth}{1pt}
\end{quote}

By the properties of the strongly-universal-2 hash $H$, if $Q_i = Q_i'$ then $\Omega_1 = \Omega_2$. Otherwise $\Omega_1 \neq \Omega_2$ with probability exponentially close to one. This produces the result as desired. The simulators then feed the transcripts to $V_0$, and terminates simulation.

\subsection{Entangled Simulators}\label{SEC:entangledsim}

The binding condition of commitment used above (protocol \ref{MIP_CHSH}) can be broken given $\mathbf{PR}$-boxes. However, if the verifier were willing to tolerate approximately $15\%$ of errors in the provers' unveiling string ($z_{1}$ or $0$), then it is possible to break binding with shared entanglement \cite{10.1007/978-3-540-45078-8_1} while maintaining soundness against local provers. Using this weakened version of commitment in place of protocol 
\ref{MIP_CHSH} yields a $\mathbf{ZK}^{\overset{\text{\tiny{poly}}}{\mathbb {|LOC\rangle}}}\mathbf{MIP}^{\emptyset}_{\emptyset}$ protocol for a
$\mathbf{NEXP}$-complete language ($\mathbf{ZK}^{\overset{\text{\tiny{poly}}}{\mathbb {|LOC\rangle}}}$ denotes shared entanglement for the simulator; consult Appendix \ref{AppB} for more notations).
We leave the details of this modification to the reader.

\section{Conclusions and Future Work}\label{SEC:conclusions}

$\mathbf{MIP}$ is cryptographic. $\mathbf{NEXP}$ is complexity theoretic. Although there exists a non-adaptive MIP which accepts NEXP (resolving the complexity of $\mathbf{MIP}$ and avoiding contamination), there seems to be a bit of an unexplored dimension on the zero-knowledge (cryptographic) side of things. LE-MIPs accomplishes two things: it makes explicit that non-adaptive verifiers are not necessary to avoid contamination, and it induces the question of non-locality with respect to zero-knowledge.
We close with four open questions. 

First, although the provers and verifiers of protocol \ref{IsolatingZKMIP} are local, the simulators are not -- they use PR-boxes. We do not know whether it is possible to simulate protocol \ref{IsolatingZKMIP} with \emph{local} simulators. In fact, we conjecture that there does not exist a $\mathbf{ZK}^{\emptyset}\mathbf{MIP}^{\emptyset}_{\emptyset}$ protocol for any $\mathbf{NEXP}$-complete language.


Second, as we have sketched out in section \ref{SEC:entangledsim}, by weakening the commitment scheme used, we get $\mathbf{ZK}^{\overset{\text{\tiny{poly}}}{\mathbb {|LOC\rangle}}}\mathbf{MIP}^{\emptyset}_{\emptyset} = \mathbf{NEXP}$. What is a minimal $\widehat{S}$ such that $\mathbf{ZK}^{\widehat{S}}\mathbf{MIP}^{\emptyset}_{\emptyset} = \mathbf{NEXP}$?

Third, as of the time of this writing, it is an open question whether $\mathbf{NEXP} \subsetneq \mathbf{MIP^*}$ \cite{IV12}. Under the locality-explicit setup, we ask a slightly more general question: does there exist a correlator $\widehat{P}$ and a corresponding LE-MIP which accepts a language $\notin \mathbf{NEXP}$? We remind the reader that characterizing the complexity classes of MIPs where the provers have non-local resources are generally open questions.

Finally, although the verifier's contamination is undesirable (in the standard MIP model), is it possible to turn it into a resource? For example, given local provers, let the verifier provide them with some non-local resources, such PR-boxes or entanglement that can be simulated in polynomial-time. This can be seen as ``enforceable honest non-local resources.'' Malicious provers would not be able to use these resources at will. Perhaps this concept would be useful in the design of multi-prover protocols.


\section*{Acknowledgements}
We would like to thank
G.~Brassard,
A.~Chailloux,
S.~Fehr,
J.~Kilian,
S.~Laplante,
J.~Li,
A.~Leverrier,
A.~Massenet,
S.~Ranellucci,
L.~Salvail,
C.~Schaffner,
and
T.~Vidick
for various discussions about earlier versions of this work. We would also like to thank Jeremy Clark for his insightful comments. Finally, we are grateful to Raphael Phan and Moti Yung for inviting us to publish a lead-up paper to this work as an {\em Insight Paper} at MyCrypt 2016.

\bibliographystyle{ieeetr}


\newcommand{\RGRB}[0]{ \text{$\mathbf{R  \! \frac{GR}{BG} \! B}$} }
\appendix

\section{Babai, Fortnow and Lund's MIP for Languages in NEXP}\label{SEC:appendix_A}

This section describes a variant of the multi-prover protocol for oracle-3-SAT found in \cite{BFL90}. We refer to this as the BFL protocol, or BFL classic.
\begin{definition}

Let $r, s > 0$ be integers. Let $z, b_1, b_2, b_3$ be strings of variables, where $|z| = r$ and $|b_i| = s$. Let $B(z, b_1, b_2, b_3, t_1, t_2, t_3)$ be a Boolean formula in $r+3s+3$ variables. A Boolean function $A : \{0, 1\}^s \rightarrow \{0, 1\}$ is a \textit{3-satisfying oracle} for $B$ if $$B(z, b_1, b_2, b_3, A(b_1), A(b_2), A(b_3)) = 1$$ for every string $z, b_1, b_2, b_3$.

$B$ is \textit{oracle-3-satisfiable} if such a function $A$ exists.

The \textit{Oracle-3-SAT} problem $(B, r, s)$ asks whether a Boolean formula $B$ is oracle-3-satisfiable, where $r$ and $s$ denote the lengths of $z$ and $b_i$, as above.

\end{definition}

\begin{lemma}

Oracle-3-SAT is $\mathbf{NEXP}$-complete.

\end{lemma}

\begin{definition}
Let $\mathbb{F}$ be an arbitrary field.
Let $\phi : \{0, 1\}^m \rightarrow \{0, 1 \}$ be a Boolean function. An \emph{arithmetization} of $\phi$ is a polynomial $f(x_1,\ldots,x_m) \in \mathbb{F}[X_1, \ldots, X_m]$ such that for all $z \in \{0, 1\}^m$, $\phi(z) = 0 \Leftrightarrow f(z) = 0$. A specific one is given in \cite{BFL90}, proposition 3.1 .\\

Equivalently, the $\phi(z) = 0 \Leftrightarrow f(z) = 0$ condition can be replaced with $\phi(z) = 1 \Leftrightarrow f(z) = 0$.

\end{definition}

\begin{quote}
\rule{\linewidth}{1pt}
\begin{protocol}\label{PROT:BFL_sumcheck} ( Sumcheck Protocol ) \end{protocol}

Let $\phi(x_1, \ldots , x_m)$ be the 3-CNF formula which the prover $P$ is trying to show to be a tautology to a verifier $V$.
Let $\mathbb{F}$ be a field of sufficient size (of order at least $(3c+1)m$ will suffice where $c$ is the number of clauses of $\phi$).

\begin{enumerate}

\item $V$ takes $\phi$ and computes its arithmetization $f$ according to \cite{BFL90} Proposition 3.1 and sends it to $P$.

\item $V$ and $P$ agree on a set $I \subset \mathbb{F}$ of size at least $2dm$ where $d$ is the degree of $f$.

\item $V$ assigns $b_0 = 0$, which is supposed to be equal to the sum $$\sum_{x_1 = 0}^1 \ldots \sum_{x_m = 0}^1 f(x_1, \ldots , x_m)^2 = 0$$

\item $i \leftarrow 1$.

\item $P$ sends the coefficients of the univariate polynomial in $x$, $$g_i(x) = h(r_1, \ldots , r_{i-1}, x) = \sum_{x_{i+1}=0}^1 \ldots \sum_{x_{m}=0}^1 f(r_1, \ldots, r_{i-1}, x, x_{i+1}, \ldots, x_m)^2$$

\item $V$ checks whether $b_{i-1} = g_i(0) + g_i(1)$. If not, abort.

\item $V$ chooses a random $r_i \in I$, computes $b_i = g_i(r_i)$ and sends $r_i$ to $P$.

\item If $i \leq m$ then $i \leftarrow i+1$ and go to step 4.

\item $V$ checks whether $b_m = f(r_1, \ldots , r_m)^2$.

\end{enumerate}

\rule{\linewidth}{1pt}
\end{quote}


%
\begin{quote}
\rule{\linewidth}{1pt}
\begin{protocol}\label{PROT:BFL_appendix} ( Babai, Fortnow and Lund's MIP for Oracle-3-SAT ) \end{protocol}

Given $(B, r, s)$ as common input.

\begin{enumerate}


\item (sumcheck with oracle) $V$ and $P_1$ execute protocol \ref{PROT:BFL_sumcheck}. Let $( Q_{k+1},Q_{k+2},Q_{k+3} ) = ( r_{r+1}...r_{r+s},r_{r+s+1}...r_{r+2s},r_{r+2s+1}...r_{r+3s} ) \in (I^{s})^{3}$ be $V$'s questions during this phase.

\item (multilinearity test) $V$ asks $P_1$ to simulate an oracle storing the function $A$. $V$ queries $P_1$ with random, linearly related values in $I^{s}$.
If any response does not satisfy linearity, abort protocol. Let $Q_1,\ldots,Q_k \in I^{s}$ be $V$'s questions during this phase.

\item (non-adaptiveness test) $V$ chooses uniformly at random an $i$ such that $1 \leq i \leq k+3$ and asks $Q_i$ to $P_2$. If $P_2$'s answer differs from that of $P_1$, reject. Otherwise accept.

\end{enumerate}

\rule{\linewidth}{1pt}
\end{quote}

\section{Non-Locality -- an introduction}\label{AppB}
In this section we solely focus on the two-party single-round games and strategies that are sufficient
to discuss and analyze most of the MIPs.
Definitions and proofs for complete generalizations to multi-party multi-round games and strategies
will appear in a forthcoming paper with co-author Adel Magra.

\subsubsection{Games:}

Let $V$ be a predicate on $A\times B\times X\times Y$ (for some finite sets $A, B, X,$ and $Y$) and let $\pi$ be a probability distribution on $A\times B$.
Then $V$ and $\pi$ define a (single-round) game $G$ as follows: A pair of questions $(a,b)$ is randomly chosen according to distribution $\pi$, and $a\in A$ is sent to Alice
and $b\in B$ is sent to Bob. Alice must respond with an answer $x\in X$ and Bob with an answer $y\in Y$.
Alice and Bob win if $V$ evaluates to 1 on $(a,b,x,y)$ and lose otherwise.

\subsubsection{Strategies: Two-Party Channels}

A strategy for Alice and Bob is simply a probability distribution $P_{(x,y | a,b)}$ describing exactly how they will answer $(x,y)$ on every
pair of questions $(a,b)$. We now breakdown the set of all possible strategies for Alice and Bob according to their {\em non-locality}.

\subsubsection{Deterministic and Local Strategies:} A strategy $P_{(x,y | a,b)}$ is {\em deterministic} if there exists functions $f_{A}:A\rightarrow X, f_{B}:B\rightarrow Y$ such that
$$P_{(x,y | a,b)} = 
\begin{cases}
    1  & \text{if $x=f_{A}(a)$ and $y=f_{B}(b)$} \\
    0  & \text{otherwise}
\end{cases}.$$
A deterministic strategy corresponds to the situation where Alice and Bob agree on their individual actions before any knowledge of the values $a,b$ is provided to them.
In this case they use only their own input to determine their individual output.


A strategy $P_{(x,y | a,b)}$ is {\em local} if there exists a finite set $R$ and functions $f_{A}:A\times R\rightarrow X, f_{B}:B\times R\rightarrow Y$ such that
$$P_{(x,y | a,b)} = \frac{ |\{ r\in R : x=f_{A}(a,r) \text{ and } y=f_{B}(b,r) | }{ |R|}.$$
A local strategy corresponds to the situation where Alice and Bob agree on a deterministic strategy selected uniformly among $|R|$ such possibilities.
The choice $r$ of Alice and Bob's strategy, and the choice of inputs $(a,b)$ provided to Alice and Bob are generally agreed to be statistically independent random variables.

\subsection{Local Reducibility}
We now turn to the notion of locally reducing a strategy to another, that is how Alice and Bob limited to local strategies but equipped
with a particular (not necessarily local) strategy $U'$ are able to achieve another particular (not necessarily local) strategy $U$. For this purpose
we introduce a notion of distance between strategies in order to analyze strategies that are approaching each other asymptotically.

\subsubsection{Distances between Strategies:} Several distances could be selected here as long as their meaning as it approaches zero are the same. In the definitions below,
$U,U'$ are strategies and ${\mathcal U'}$ is a finite set of strategies. 

\begin{definition}
$ |U,U'| = \displaystyle{ \sum_{a,b,x,y} } | P_{U}(x,y | a,b) - P_{U'}(x,y | a,b) | $
\end{definition}

\begin{definition}
$ |U,{\mathcal U'}| =\displaystyle{ \min_{U'\in {\mathcal U'} } } |U,U'| $
\end{definition}

\subsubsection{Local extensions of Strategies:} For natural integer $n$, we define the set $\operatorname{LOC}^{n}(U)$ of strategies that are local extensions (of order $n$)
of $U$ to be all the strategies Alice and Bob can achieve using local strategies where strategy $U$ may be used up to $n$ times as sub-routine
calls\footnote{Done by selecting functions $f_{A}^{0}:A\times R\rightarrow A, ~f_{A}^{1}:A\times X \times R\rightarrow A,..., ~f_{A}^{n-1}:A \times X^{n-1}\times R\rightarrow A$, $ ~f_{A}^{n}:A \times X^{n}\times R\rightarrow X$ to determine the input of each sub-routine from input $a$ and previous outputs.}. If we restrict all the functions used to be polynomial-time computable we analogously define $\underset{\text{\small{poly}}}{\operatorname{LOC}}^{n}(U)$.

\begin{definition}
$U'$ Locally (poly-)Reduces to $U$ ($U' \leq_{\underset{\text{\tiny{(poly)}}}{\mbox{\scriptsize\em LOC}}} U\!$) iff $\displaystyle{\lim_{n\rightarrow \infty}} | U' , \underset{\text{\small{(poly)}}}{\mbox{\em LOC}}^{n}(U) | =0$.
\end{definition}

\begin{definition}
$U'$ is Locally (poly-)Equivalent to $U$ ($U' =_{\underset{\text{\tiny{(poly)}}}{\mbox{\scriptsize\em LOC}}} U\!$) iff
$U' \leq_{\underset{\text{\tiny{(poly)}}}{\mbox{\scriptsize\em LOC}}} U \leq_{\underset{\text{\tiny{(poly)}}}{\mbox{\scriptsize\em LOC}}} U'.$
\end{definition}

\subsubsection{Non-Adaptive extensions of Strategies:} For natural integer $n$, we define the set $\operatorname{NAD}^{n}(U)$ of strategies that are Non-Adaptive extensions (of order $n$)
of $U$ to be all the strategies Alice and Bob can achieve using Non-Adaptive strategies where strategy $U$ may be used up to $n$ times as sub-routine
calls\footnote{Done by selecting functions $f_{A}^{0}:A\times R\rightarrow A, ~f_{A}^{1}:A \times R\rightarrow A,..., ~f_{A}^{n-1}:A \times R\rightarrow A$,\\ $ ~f_{A}^{n}:A \times X^{n}\times R\rightarrow X$ to determine the input of each sub-routine from input $a$ only.}. If we restrict the functions used to be poly-time computable we get $\underset{\text{\small{poly}}}{\operatorname{NAD}}^{n}(U)$.

\begin{definition}
$U'$ Non-Adaptively (poly-)Reduces to $U$ ($U' \leq_{\underset{\text{\tiny{(poly)}}}{\mbox{\scriptsize\em NAD}}} U\!$) iff $\displaystyle{\lim_{n\rightarrow \infty}} | U' , \underset{\text{\small{(poly)}}}{\mbox{\em NAD}}^{n}(U) | =0$.
\end{definition}

\begin{definition}
$U'$ is Non-Adaptively (poly-)Equivalent to $U$ ($U' =_{\underset{\text{\tiny{(poly)}}}{\mbox{\scriptsize\em NAD}}} U\!$) iff 
$U' \leq_{\underset{\text{\tiny{(poly)}}}{\mbox{\scriptsize\em NAD}}} U \leq_{\underset{\text{\tiny{(poly)}}}{\mbox{\scriptsize\em NAD}}} U'.$
\end{definition}

In general, Non-Adaptive reducibility is a weaker notion than local reducibility. However, for certain distributions $\mathbf{U}$ it may result that $\{ D | D \leq_{\underset{\text{\tiny{(poly)}}}{\mbox{\scriptsize\em LOC}}} \mathbf{U} \}
= \{ D | D \leq_{\underset{\text{\tiny{(poly)}}}{\mbox{\scriptsize\em NAD}}} \mathbf{U} \}$ as follows.

\subsection{Locality}
We now define the lowest of the non-locality classes ${\mathbb{LOC}}$. We could define it directly from the notion of local strategies as defined above, but for analogy with the other classes
we later define, ${\mathbb{LOC}}$ is defined as all those strategies locally reducible to a {\em complete} strategy we call $\mathbf{ID}$ (see {\bf Fig.~\ref{ID}}). Of course, any strategy is complete for this class.
\begin{figure}[h!]

\centering
  \mbox{\Qcircuit @C=1em @R=.7em {
      \lstick{a}  \ar[r] & \multigate{1}{\mathbf{ID}} & \rstick{b} \ar[l] \\
      \lstick{a}   & \ghost{\mathbf{ID}} \ar[r]\ar[l] & \rstick{b}
    }}
\caption{an $\mathbf{ID}$-box}
  \label{ID}
\end{figure}

\begin{definition}
${\mathbb{LOC}} = \{ U | U \leq_{\mbox{\scriptsize\em LOC}} \mathbf{ID} \}$ and $\underset{{{poly}}}{\mathbb{LOC}} = \{ U | U \leq_{\underset{\text{\tiny{poly}}}{\mbox{\scriptsize\em LOC}}} \mathbf{ID} \}$
\end{definition}

Note: ${\mathbb{LOC}}$ is the class of strategies that John Bell \cite{BELL64} considered as classical hidden-variable theories that he compared to entanglement.
It is also the class of strategies that BenOr, Goldwasser, Kilian and Wigderson \cite{BGKW88} chose to define classical Provers in Multi-Provers Interactive Proof Systems.
 ${\mathbb{LOC}}$ is also those strategies Non-Adaptively reducible to $\mathbf{ID}$

\begin{definition}
Alternatively, ${\mathbb{LOC}} = \{ U | U \leq_{\mbox{\scriptsize\em NAD}} \mathbf{ID} \}$ and $\underset{{{poly}}}{\mathbb{LOC}} = \{ U | U \leq_{\underset{\text{\tiny{poly}}}{\mbox{\scriptsize\em NAD}}} \mathbf{ID} \}$
\end{definition}

Alternatively, we can also define $\mathbb{LOC}$ from an empty box as used in the core of this paper
\begin{figure}[h!]

\centering
  \mbox{\Qcircuit @C=1em @R=.7em {
      \lstick{a}  \ar[r] & \multigate{1}{\mathbf{\emptyset}} & \rstick{b} \ar[l] \\
      \lstick{x}   & \ghost{\mathbf{\emptyset}} \ar[r]\ar[l] & \rstick{y}
    }}
\caption{an $\mathbf{\emptyset}$-box where $x\in X$ and $y\in Y$ are uniform and independent of everything else}
  \label{PHI}
\end{figure}
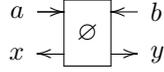

\begin{definition}
Alternatively, ${\mathbb{LOC}} = \{ U | U \leq_{\mbox{\scriptsize\em NAD}} \mathbf{\emptyset} \} = \{ U | U \leq_{\mbox{\scriptsize\em LOC}} \mathbf{\emptyset} \} $
\end{definition}

\subsection{One-Way Signalling}
We now turn to One-Way Signalling which allows communication from one side to the other. We name the directions arbitrarily Left and Right.
We define ${\mathbf R} \text{-} {\mathbb{SIG}}$ (resp. ${\mathbf L} \text{-} {\mathbb{SIG}}$) as all those strategies locally reducible to a {\em complete} strategy we call ${\mathbf R} \text{-} {\mathbf{SIG}}$ (see {\bf Fig.~\ref{rSIG}})
(resp. ${\mathbf L} \text{-} {\mathbf{SIG}}$ (see {\bf Fig.~\ref{lSIG}})). These classes are useful to define what it means for a strategy to {\em signal} as well as the notion of {\em No-Signalling} strategies.

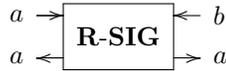
\begin{figure}[h!]

\centering
  \mbox{\Qcircuit @C=1em @R=.7em {
      \lstick{a}  \ar[r] & \multigate{1}{ {\mathbf R} \text{-} {\mathbf{SIG}} } & \rstick{b} \ar[l] \\
      \lstick{a}   & \ghost{ {\mathbf R} \text{-} {\mathbf{SIG}} } \ar[r]\ar[l] & \rstick{a}
    }}
\caption{an $\mathbf{R} \text{-} {\mathbf{SIG}}$-box}
  \label{rSIG}
\end{figure}

\begin{definition}
${\mathbf R} \text{-} {\mathbb{SIG}} = \{ U | U \leq_{\mbox{\scriptsize\em LOC}} {\mathbf R} \text{-} {\mathbf{SIG}} \}$ and ${\mathbf R} \text{-} \underset{{{poly}}}{\mathbb{SIG}} = \{ U | U \leq_{\underset{\text{\tiny{poly}}}{\mbox{\scriptsize\em LOC}}} {\mathbf R} \text{-} {\mathbf{SIG}} \}$
\end{definition}

\begin{definition}
We say that $U$ Right Signals (is ${\mathbf R} \text{-} {\mathbb{SIG}}$-verbose\footnote{We define the notion of $\mathbb{L}$-verbose in analogy to $\mathbb{NP}$-hard: it means ``as verbose as any distribution in non-locality class
$\mathbb{L}$''. In consequence, a distribution $U$ is $\mathbb{L}$-complete if $U \in \mathbb{L}$ and $U$ is $\mathbb{L}$-verbose.}) iff ${\mathbf R} \text{-} {\mathbf{SIG}} \leq_{\mbox{\scriptsize\em LOC}} U$.
\end{definition}

\begin{figure}[h!]

\centering
  \mbox{\Qcircuit @C=1em @R=.7em {
      \lstick{a}  \ar[r] & \multigate{1}{{\mathbf L} \text{-} {\mathbf{SIG}}} & \rstick{b} \ar[l] \\
      \lstick{b}   & \ghost{{\mathbf L} \text{-} {\mathbf{SIG}}} \ar[r]\ar[l] & \rstick{b}
    }}
\caption{an ${\mathbf L} \text{-} {\mathbf{SIG}}$-box}
  \label{lSIG}
\end{figure}

\begin{definition}
${\mathbf L} \text{-} {\mathbb{SIG}} = \{ U | U \leq_{\mbox{\scriptsize\em LOC}} {\mathbf L} \text{-} {\mathbf{SIG}} \}$ and ${\mathbf L} \text{-} \underset{{{poly}}}{\mathbb{SIG}} = \{ U | U \leq_{\underset{\text{\tiny{poly}}}{\mbox{\scriptsize\em LOC}}} {\mathbf L} \text{-} {\mathbf{SIG}} \}$

\end{definition}

\begin{definition}
We say that $U$ Left Signals (is ${\mathbf L} \text{-} {\mathbb{SIG}}$-verbose) iff ${\mathbf L} \text{-} {\mathbf{SIG}} \leq_{\mbox{\scriptsize\em LOC}} U$.
\end{definition}

\begin{definition}
We say that $U$ Signals iff $U$ Right Signals or Left Signals.
\end{definition}

We prove a first result that is intuitively obvious. We show that the complete strategy ${\mathbf R} \text{-} {\mathbf{SIG}}$ cannot be approximated in ${\mathbf L} \text{-} {\mathbb{SIG}}$
and the other way around.

\begin{theorem}\label{LR-impossible}
${\mathbf R} \text{-} {\mathbf{SIG}} \not\in {\mathbf L} \text{-} {\mathbb{SIG}}$ and
${\mathbf L} \text{-} {\mathbf{SIG}} \not\in {\mathbf R} \text{-} {\mathbb{SIG}}$.
\end{theorem}

\begin{proof}
Follows from a simple capacity argument. For all $n$, all the channels in $\mbox{LOC}^{n}({\mathbf R} \text{-} {\mathbf{SIG}})$ have zero left-capacity, while ${\mathbf L} \text{-} {\mathbf{SIG}}$ has non-zero left-capacity. And vice-versa.
\end{proof}

\subsection{Signalling}
We are now ready to define the largest of the non-locality classes ${\mathbb{SIG}}$. Indeed every possible strategy is in ${\mathbb{SIG}}$.

\begin{definition}

${\mathbb{SIG}} = \{ U | U \leq_{\mbox{\scriptsize\em LOC}} \mathbf{SIG}\}$ and $\underset{{{poly}}}{\mathbb{SIG}} = \{ U | U \leq_{\underset{\text{\tiny{poly}}}{\mbox{\scriptsize\em LOC}}} {\mathbf{SIG}} \}$

\end{definition}

\begin{figure}[h!]

\centering
  \mbox{\Qcircuit @C=1em @R=.7em {
      \lstick{a}  \ar[r] & \multigate{1}{{\mathbf{SIG}}} & \rstick{b} \ar[l] \\
      \lstick{b}   & \ghost{{\mathbf{SIG}}} \ar[r]\ar[l] & \rstick{a}
    }}
\caption{a ${\mathbf{SIG}}$-box}
  \label{SIG}
\end{figure}

\begin{definition}
We say that $U$ Fully Signals (is ${\mathbb{SIG}}$-verbose) iff $U$ Right Signals and Left Signals.
\end{definition}

\begin{figure}[hbt]
\begin{center}
\fbox{
\includegraphics[width=1.0\textwidth]{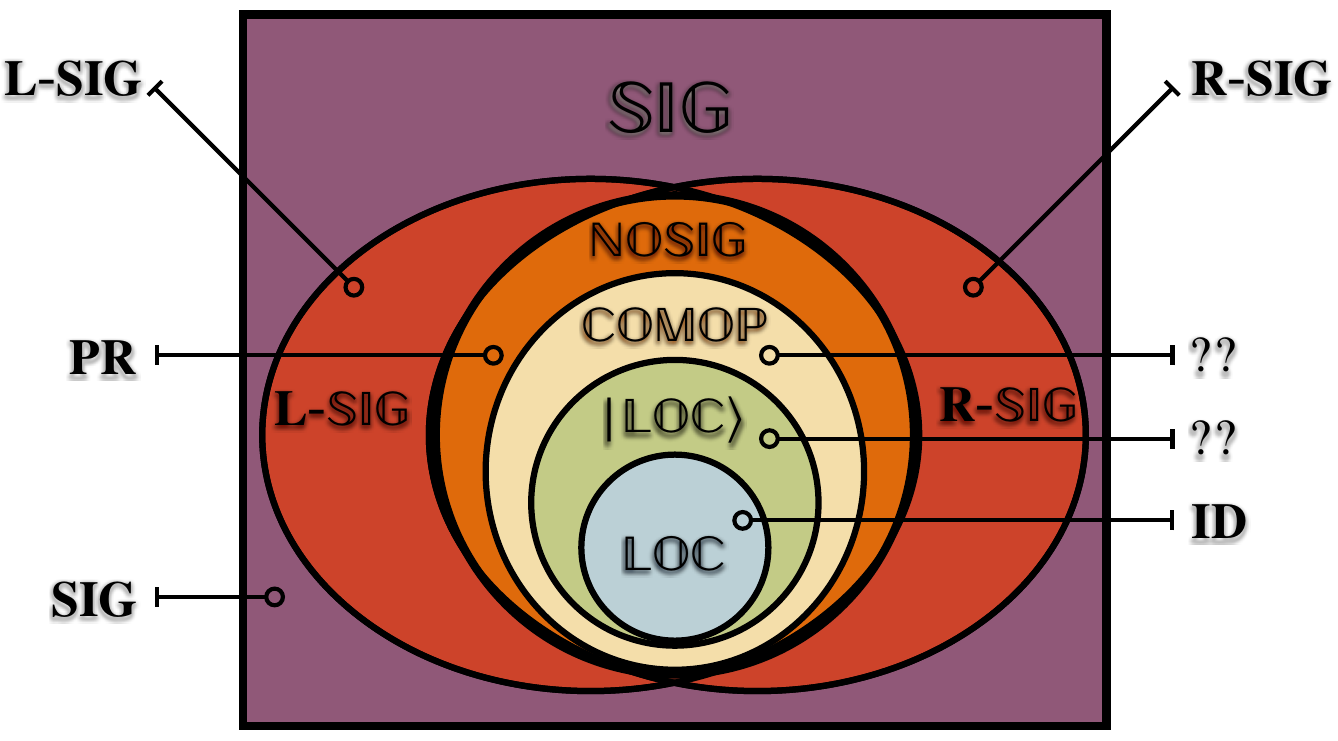}
}
\end{center}
\caption{\label{HIER} Non-locality Hierarchy and complete (two-party) distributions in each class.}
\end{figure}

\subsection{No-Signalling}
We finally define the less intuitive non-locality class ${\mathbb{NOSIG}}$ in relation to classes defined above. 

\begin{definition}
$\mathbb{NOSIG} = {\mathbf R} \text{-} {\mathbb{SIG}} \bigcap {\mathbf L} \text{-} {\mathbb{SIG}}$ and $\underset{{{poly}}}{\mathbb{NOSIG}} = {\mathbf R} \text{-} \underset{{{poly}}}{\mathbb{SIG}} \bigcap {\mathbf L} \text{-} \underset{{{poly}}}{\mathbb{SIG}}$.
\end{definition}

A similar characterization may be found in \cite{Acin2015} Section 3 and \cite{Barnum05} Corollary 3.5.

\begin{theorem} \label{NSig}.
The above definition of $\mathbb{NOSIG}$ exactly coincides with the {\em traditional} notion of No-Signalling \cite{BLM+05}.
\end{theorem}

Intuitively, a distribution $P(x,y | a,b)$ is No-Signalling as long as for every $a$ the $x|b$ and for every $b$ the $y|a$  channels have zero capacity.
 
Note: Forster and Wolf \cite{PhysRevA.84.042112} have proved that {\bf PR} (see {\bf Fig.~\ref{nlbox}}) is complete for $\mathbb{NOSIG}$ distributions under an asymptotic definition similar to ours.

{\bf Fig.~\ref{HIER}} shows the relation of these classes as well as the case obtained via quantum entanglement (${\mathbb{|LOC\rangle }}$) as considered by Bell \cite{BELL64} and via commuting-operators (${\mathbb{COMOP}}$) as defined by 
Ito, 
Kobayashi, 
Preda, 
Sun, and 
Yao \cite{DBLP:conf/coco/ItoKPSY08}. We include those for completeness but will not discuss these particular classes any further in this work.

\begin{definition}
We say that $U$ does not Signal iff $U$ does not Right Signal nor Left Signal iff $U \in \mathbb{NOSIG}$.
\end{definition}

\section{Visual description of the new model}

\subsection{Local Multi-Prover Interactive Proofs}

%
%

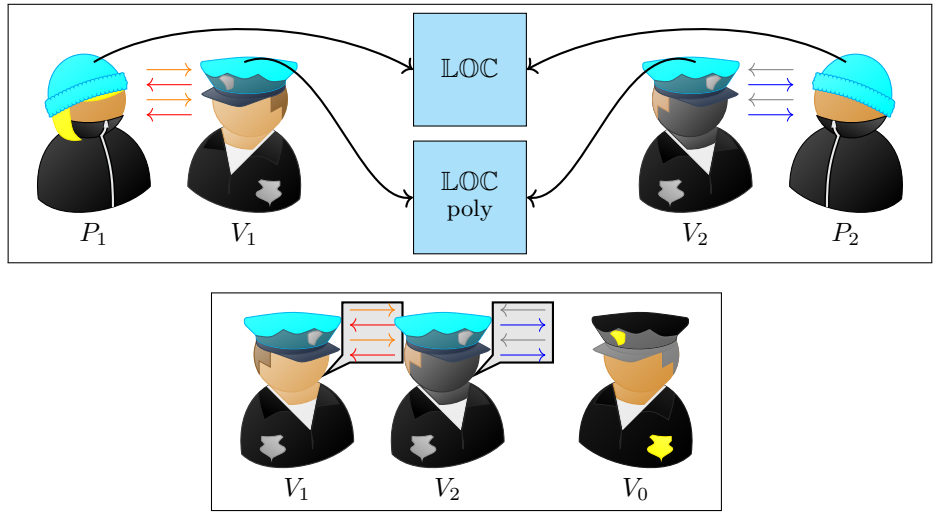
\begin{figure}[htb]
\begin{center}
\fbox{
\begin{tikzpicture}
\node[draw,rectangle,color=black,fill=cyan!30,minimum size=1.5cm] (U) at (5,-0.7) { $\underset{\text{\small{poly}}}{\mathbb {LOC}}$ };
\node[draw,rectangle,color=black,fill=cyan!30,minimum size=1.5cm] (V) at (5,1) { ${\mathbb {LOC}}$ };
\node[criminal,saturated,female,shirt=black,hat=cyan,hair=yellow,minimum size=1.5cm] (A) at (0,0.1) {$P_{1}$};
\node[police,mirrored,shirt=black,hat=cyan,hatbadge=gray,badge=gray,minimum size=1.5cm] (B) at (2,0.1) {$V_{1}$};
\node[police,shirt=black,hat=cyan,skin=black,hair=brown,hatbadge=gray,badge=gray,minimum size=1.5cm] (C) at (8,0.1) {$V_{2}$};
\node[criminal,saturated,mirrored,shirt=black,hat=cyan,minimum size=1.5cm] (D) at (10,0.1) {$P_{2}$};
\draw[orange, ->] (0.7,1) -- (1.3,1);
\draw[red, <-] (0.7,0.8) -- (1.3,0.8);
\draw[orange, ->] (0.7,0.6) -- (1.3,0.6);
\draw[red, <-] (0.7,0.4) -- (1.3,0.4);
\draw[gray, ->] (9.3,1) -- (8.7,1);
\draw[blue, <-] (9.3,0.8) -- (8.7,0.8);
\draw[gray, ->] (9.3,0.6) -- (8.7,0.6);
\draw[blue, <-] (9.3,0.4) -- (8.7,0.4);
\draw[black, thick, ->] (A.north) .. controls (1,1.7)  and (3,1.7)  .. (V.west);
\draw[black, thick, ->] (D.north) .. controls (9,1.7)  and (7,1.7)  .. (V.east);
\draw[black, thick, ->] (B.north)  .. controls (3,1.5) and (3.5,-0.8) .. (U.west);
\draw[black, thick, ->] (C.north)  .. controls (7,1.5) and (6.5,-0.8) ..  (U.east);
\end{tikzpicture}
}
\end{center}
%
\begin{center}
\fbox{
\begin{tikzpicture}
\node[draw,rectangle,color=white,minimum size=1.5cm] at (5,-0.7) {};
\node[police,shirt=black,hat=cyan,badge=gray,hatbadge=gray,minimum size=1.5cm] (A) at (0.5,0) {$V_{1}$};
\draw[black,fill=black!10, thick] (1.1,0.5) -- (1.1,1.2) -- (1.9,1.2) -- (1.9,0.4) -- (1.2,0.4) -- (A.mouth) -- cycle;
\draw[orange, ->] (1.2,1.1) -- (1.8,1.1);
\draw[red, <-] (1.2,0.9) -- (1.8,0.9);
\draw[orange, ->] (1.2,0.7) -- (1.8,0.7);
\draw[red, <-] (1.2,0.5) -- (1.8,0.5);
\node[police,saturated,shirt=black,hat=black,hatshield=gray,mirrored,hair=gray,hatbadge=yellow,badge=yellow,minimum size=1.5cm] (B) at (5,0) {$V_{0}$};
\node[police,shirt=black,hat=cyan,skin=black,hair=brown,hatbadge=gray,badge=gray,minimum size=1.5cm] (C) at (2.5,0) {$V_{2}$};
\draw[black,fill=black!10, thick] (3.1,0.5) -- (3.1,1.2) -- (3.9,1.2) -- (3.9,0.4) -- (3.2,0.4) -- (C.mouth) -- cycle;
\draw[gray, ->] (3.8,1.1) -- (3.2,1.1);
\draw[blue, <-] (3.8,0.9) -- (3.2,0.9);
\draw[gray, ->] (3.8,0.7) -- (3.2,0.7);
\draw[blue, <-] (3.8,0.5) -- (3.2,0.5);
\end{tikzpicture}
}
\end{center}
\caption{\label{inter} Interrogation phase (top) followed by decision phase (bottom).}
\end{figure}

In the Interrogation phase (see {\bf Fig. \ref{inter}}) $V_{1},...,V_{k}$ (equipped with an arbitrary local correlator) individually interrogate $P_{1},...,P_{k}$ (equipped with an arbitrary local correlator).
At the end of the interactive part, all the $V_{1},...,V_{k}$ report to $V_{0}$ who takes the final decision. The corresponding complexity class is $\mathbf{MIP} = \mathbf{MIP}^{{\mathbb {LOC}}}_{\underset{\text{\tiny{poly}}}{\mathbb {LOC}}} = \mathbf{NEXP}$.

%

\subsection{Entangled Multi-Prover Interactive Proofs}
\label{ent-def}

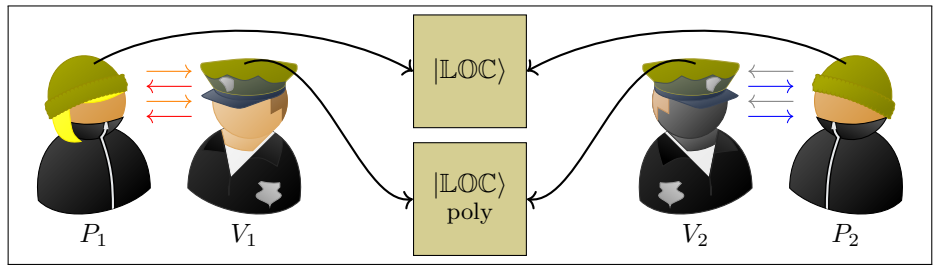
\begin{figure}[htb]
\begin{center}
\fbox{
\begin{tikzpicture}
\node[draw,rectangle,color=black,fill=olive!40,minimum size=1.5cm] (U) at (5,-0.7) {$\underset{\text{\small{poly}}}{|{\mathbb {LOC}}\rangle}$};
\node[draw,rectangle,color=black,fill=olive!40,minimum size=1.5cm] (V) at (5,1) {$|{\mathbb {LOC}}\rangle$};
\node[criminal,saturated,female,shirt=black,hat=olive,hair=yellow,minimum size=1.5cm] (A) at (0,0.1) {$P_{1}$};
\node[police,mirrored,shirt=black,hat=olive,hatbadge=gray,badge=gray,minimum size=1.5cm] (B) at (2,0.1) {$V_{1}$};
\node[police,shirt=black,hat=olive,skin=black,hair=brown,hatbadge=gray,badge=gray,minimum size=1.5cm] (C) at (8,0.1) {$V_{2}$};
\node[criminal,saturated,mirrored,shirt=black,hat=olive,minimum size=1.5cm] (D) at (10,0.1) {$P_{2}$};
\draw[orange, ->] (0.7,1) -- (1.3,1);
\draw[red, <-] (0.7,0.8) -- (1.3,0.8);
\draw[orange, ->] (0.7,0.6) -- (1.3,0.6);
\draw[red, <-] (0.7,0.4) -- (1.3,0.4);
\draw[gray, ->] (9.3,1) -- (8.7,1);
\draw[blue, <-] (9.3,0.8) -- (8.7,0.8);
\draw[gray, ->] (9.3,0.6) -- (8.7,0.6);
\draw[blue, <-] (9.3,0.4) -- (8.7,0.4);
\draw[black, thick, ->] (A.north) .. controls (1,1.7)  and (3,1.7)  .. (V.west);
\draw[black, thick, ->] (D.north) .. controls (9,1.7)  and (7,1.7)  .. (V.east);
\draw[black, thick, ->] (B.north)  .. controls (3,1.5) and (3.5,-0.8) .. (U.west);
\draw[black, thick, ->] (C.north)  .. controls (7,1.5) and (6.5,-0.8) ..  (U.east);
\end{tikzpicture}
}
\end{center}
\caption{\label{interQ} Interrogation phase.}
\end{figure}

In the Interrogation phase (see {\bf Fig. \ref{interQ}}) $V_{1},...,V_{k}$ (equipped with an arbitrary entangled correlator) individually interrogate $P_{1},...,P_{k}$ (equipped with an arbitrary entangled correlator).
At the end of the interactive part, all the $V_{1},...,V_{k}$ report to $V_{0}$ who takes the final decision.
The corresponding complexity class is $\mathbf{MIP}^{*} = \mathbf{MIP}^{{\mathbb{| LOC \rangle }}}_{\underset{\text{\tiny{poly}}}{\mathbb {| LOC \rangle }}} \supseteq \mathbf{NEXP}$.

\subsection{No-Signalling Multi-Prover Interactive Proofs}
\label{NS-def}

%

\begin{figure}[htb]
\begin{center}
\fbox{
\begin{tikzpicture}
\node[draw,rectangle,color=black,fill=orange!90,minimum size=1.5cm] (U) at (5,-0.7) {$\underset{\text{\small{poly}}}{\mathbb {NOSIG}}$};
\node[draw,rectangle,color=black,fill=orange!90,minimum size=1.5cm] (V) at (5,1) {${\mathbb {NOSIG}}$};
\node[criminal,saturated,female,shirt=black,hat=orange,hair=yellow,minimum size=1.5cm] (A) at (0,0.1) {$P_{1}$};
\node[police,mirrored,shirt=black,hat=orange,hatbadge=gray,badge=gray,minimum size=1.5cm] (B) at (2,0.1) {$V_{1}$};
\node[police,shirt=black,hat=orange,skin=black,hair=brown,hatbadge=gray,badge=gray,minimum size=1.5cm] (C) at (8,0.1) {$V_{2}$};
\node[criminal,saturated,mirrored,shirt=black,hat=orange,minimum size=1.5cm] (D) at (10,0.1) {$P_{2}$};
\draw[orange, ->] (0.7,1) -- (1.3,1);
\draw[red, <-] (0.7,0.8) -- (1.3,0.8);
\draw[orange, ->] (0.7,0.6) -- (1.3,0.6);
\draw[red, <-] (0.7,0.4) -- (1.3,0.4);
\draw[gray, ->] (9.3,1) -- (8.7,1);
\draw[blue, <-] (9.3,0.8) -- (8.7,0.8);
\draw[gray, ->] (9.3,0.6) -- (8.7,0.6);
\draw[blue, <-] (9.3,0.4) -- (8.7,0.4);
\draw[black, thick, <->] (A.north) .. controls (1,1.7)  and (3,1.7)  .. (V.west);
\draw[black, thick, <->] (D.north) .. controls (9,1.7)  and (7,1.7)  .. (V.east);
\draw[black, thick, <->] (B.north)  .. controls (3,1.5) and (3.5,-0.8) .. (U.west);
\draw[black, thick, <->] (C.north)  .. controls (7,1.5) and (6.5,-0.8) ..  (U.east);
\end{tikzpicture}
}
\end{center}
\caption{\label{NoSIG} Interrogation phase.}
\end{figure}

In the Interrogation phase (see {\bf Fig. \ref{NoSIG}}) $V_{1},...,V_{k}$ (equipped with an arbitrary No-Signalling correlator) individually interrogate $P_{1},...,P_{k}$ (equipped with an arbitrary No-Signalling correlator).
At the end of the interactive part, all the $V_{1},...,V_{k}$ report to $V_{0}$ who takes the final decision.
The corresponding complexity class is $\mathbf{MIP}^{ns} = \mathbf{MIP}^{{\mathbb {NOSIG}}}_{\underset{\text{\tiny{poly}}}{\mathbb {NOSIG}}}  = \mathbf{EXP}$.

As noted before, most MIPs found in the literature are actually (non-adaptive) local-verifier MIPs (see {\bf Fig. \ref{NoSIG-LOCAL}}) yielding for instance
$\mathbf{MIP}^{ns} = \mathbf{MIP}^{{\mathbb {NOSIG}}}_{\underset{\text{\tiny{poly}}}{\mathbb {LOC}}}$.

\begin{figure}[htb]
\begin{center}
\fbox{
\begin{tikzpicture}
\node[draw,rectangle,color=black,fill=cyan!30,text=black,minimum size=1.5cm] (U) at (5,-0.7) {$\underset{\text{\small{poly}}}{\mathbb {LOC}}$};
\node[draw,rectangle,color=black,fill=orange,minimum size=1.5cm] (V) at (5,1) {${\mathbb {NOSIG}}$};
\node[criminal,saturated,female,shirt=black,hat=orange,hair=yellow,minimum size=1.5cm] (A) at (0,0.1) {$P_{1}$};
\node[police,mirrored,shirt=black,hat=cyan,hatbadge=gray,badge=gray,minimum size=1.5cm] (B) at (2,0.1) {$V_{1}$};
\node[police,shirt=black,hat=cyan,skin=black,hair=brown,hatbadge=gray,badge=gray,minimum size=1.5cm] (C) at (8,0.1) {$V_{2}$};
\node[criminal,saturated,mirrored,shirt=black,hat=orange,minimum size=1.5cm] (D) at (10,0.1) {$P_{2}$};
\draw[orange, ->] (0.7,1) -- (1.3,1);
\draw[red, <-] (0.7,0.8) -- (1.3,0.8);
\draw[orange, ->] (0.7,0.6) -- (1.3,0.6);
\draw[red, <-] (0.7,0.4) -- (1.3,0.4);
\draw[gray, ->] (9.3,1) -- (8.7,1);
\draw[blue, <-] (9.3,0.8) -- (8.7,0.8);
\draw[gray, ->] (9.3,0.6) -- (8.7,0.6);
\draw[blue, <-] (9.3,0.4) -- (8.7,0.4);
\draw[black, thick, <->] (A.north) .. controls (1,1.7)  and (3,1.7)  .. (V.west);
\draw[black, thick, <->] (D.north) .. controls (9,1.7)  and (7,1.7)  .. (V.east);
\draw[black, thick, ->] (B.north)  .. controls (3,1.5) and (3.5,-0.8) .. (U.west);
\draw[black, thick, ->] (C.north)  .. controls (7,1.5) and (6.5,-0.8) ..  (U.east);
\end{tikzpicture}
}
\end{center}
\caption{\label{NoSIG-LOCAL} Interrogation phase.}
\end{figure}

\subsection{A New, Stronger Flavour of Zero-Knowledge}\label{sec:newzk}
Traditionally zero-knowledge is defined as a property of the honest provers for all (polynomial-time) verifiers
$$\forall_{\text{poly}} V^\prime~ \exists_{\text{poly}} S ~ \forall x\!\in\!L ~ \forall w ~~ {\mathbf{VIEW}}_{V^\prime}[ P_{1},...,P_{k},V^\prime ](w,x) = S(w,x). $$

However, in the present context, the fact that the simulation of $V^\prime$'s view via a single centralized simulator $S$, achieving zero-knowledge is rather easy
because such an $S$ can cheat the binding property of the commitments at will. The intuition behind the original definition is that the verifier is unable to convince a third party
(a Judge $J_{0}$) because the {\bf VIEW} he reports (see {\bf Fig.~\ref{dec}}) could have been equally created (with the same distribution) by a simulator.
Nevertheless, a stronger flavour of zero-knowledge is achieved if the simulator is not invoking
its full signalling power whenever the verifier does not use such power.

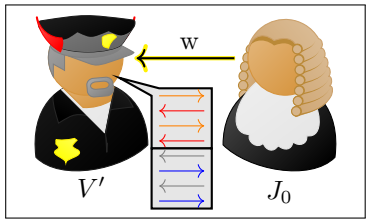
\begin{figure}[htb]
\begin{center}
\fbox{
\begin{tikzpicture}
\node[draw,rectangle,color=white,minimum size=1.5cm] at (5,-0.7) {};
\node[judge,saturated,mirrored,minimum size=1.5cm] (B) at (5,0) {$J_{0}$};
\node[police,evil,saturated,shirt=black,hat=black,hatshield=gray,hair=gray,hatbadge=yellow,badge=yellow,minimum size=1.5cm] (C) at (2.5,0) {$V^\prime$};
\draw[yellow, ultra thick, ->] (4.41, 0.5) -- (3.07,0.5);
\draw[black, thick, ->] (4.4,0.5) -- (3.1,0.5);
\filldraw[black] (3.8,0.5) circle (0pt) node[anchor=south] {w};
\draw[black,fill=black!10, thick] (3.3,0) -- (3.3,-0.7) -- (4.1,-0.7) -- (4.1,0.1) -- (3.4,0.1) -- (C.mouth) -- cycle;
\draw[orange, <-] (4.0,0) -- (3.4,0);
\draw[red, ->] (4.0,-0.2) -- (3.4,-0.2);
\draw[orange, <-] (4.0,-0.4) -- (3.4,-0.4);
\draw[red, ->] (4.0,-0.6) -- (3.4,-0.6);
\draw[black,fill=black!10, thick] (3.3,-0.7) -- (3.3,-1.5) -- (4.1,-1.5) -- (4.1,-0.7) -- cycle;
\draw[gray, ->] (4.0,-0.8) -- (3.4,-0.8);
\draw[blue, <-] (4.0,-1.0) -- (3.4,-1.0);
\draw[gray, ->] (4.0,-1.2) -- (3.4,-1.2);
\draw[blue, <-] (4.0,-1.4) -- (3.4,-1.4);

\end{tikzpicture}
} 
\end{center}
\caption{\label{dec} (Interac/Simula)tion-Distinction phase.}
\end{figure}

For all non-locality levels starting with $\widehat{{\mathbb {S}}}$ and up, the simulators $S_{i}$ do not need more non-local power than the verifiers $V^\prime_{i}$.
The ultimate (strongest) notion of ``$\underset{\text{\small{poly}}}{\mathbb {LOC}}$-local ZK'' being $\mathbf{ZK}^{\overset{\text{\tiny{poly}}}{\mathbb {LOC}}}$ because at all levels $V^\prime$ is simulated by a simulator with no extra
non-local power, whereas at the opposite end of the spectrum $\mathbf{ZK}^{\overset{\text{\tiny{poly}}}{\mathbb {SIG}}}$ is what is generally considered zero-knowledge with a single simulator or a group of signalling simulators.

This stronger notion of zero-knowledge is particularly interesting in the relativistic bit-commitment scenario where a pair of judges
may provide separate auxiliary-inputs to spatially separated verifiers pretending to be speaking to powerful provers. If the verifiers can report
their conversation fast enough to the judges (but not interact with the judges however), they must be able to do so without invoking signalling because of the distance separating them. If a pair of simulators
can produce the same distribution of views in the same context, we obtain a stronger flavour of zero-knowledge (See {\bf Fig.~\ref{inter2}}).

%

%
%


\begin{figure}[htb]
\begin{center}
\fbox{
\begin{tikzpicture}
\node[draw,rectangle,color=white,fill=white,minimum size=1.5cm] (U) at (5,-0.7) {$|{\mathbb {LOC}}\rangle$} ;
\node[draw,rectangle,color=black,fill=gray!40,minimum size=1.5cm] (V) at (5,1) { $\widehat{{\mathbb {V}}^{\prime}}$ };
\node[police,evil,shirt=black,hat=gray,hair=yellow,hatbadge=gray,badge=gray,minimum size=1.5cm] (A) at (0,0.1) {${V^{\prime}_{1}}$};
\node[police,mirrored,evil,shirt=black,hat=gray,skin=black,hair=brown,hatbadge=gray,badge=gray,minimum size=1.5cm] (D) at (10,0.1) {${V^{\prime}_{2}}$};
\node[judge,mirrored,female,hair=yellow,minimum size=1.5cm] (B) at (2.5,0.1) {$J_{1}$};
\node[judge,hair=gray,minimum size=1.5cm] (C) at (7.5,0.1) {$J_{2}$};
%
\draw[yellow, ultra thick, ->] (1.91, 0.5) -- (0.57,0.5);
\draw[black, thick, ->] (1.9,0.5) -- (0.6,0.5);
\filldraw[black] (1.3,0.5) circle (0pt) node[anchor=south] {$w_{1}$};
\draw[black,fill=black!10, thick] (0.8,0) -- (0.8,-0.7) -- (1.6,-0.7) -- (1.6,0.1) -- (0.9,0.1) -- (A.mouth) -- cycle;
\draw[orange, <-] (1.5,0) -- (0.9,0);
\draw[red, ->] (1.5,-0.2) -- (0.9,-0.2);
\draw[orange, <-] (1.5,-0.4) -- (0.9,-0.4);
\draw[red, ->] (1.5,-0.6) -- (0.9,-0.6);
\draw[yellow, ultra thick, ->] (8.09, 0.5) -- (9.43,0.5);
\draw[black, thick, ->] (8.1,0.5) -- (9.4,0.5);
\filldraw[black] (8.7,0.5) circle (0pt) node[anchor=south] {$w_{2}$};
\draw[black,fill=black!10, thick] (9.2,0) -- (9.2,-0.7) -- (8.4,-0.7) -- (8.4,0.1) -- (9.1,0.1) -- (D.mouth) -- cycle;
\draw[gray, <-] (8.5,0) -- (9.1,0);
\draw[blue, ->] (8.5,-0.2) -- (9.1,-0.2);
\draw[gray, <-] (8.5,-0.4) -- (9.1,-0.4);
\draw[blue, ->] (8.5,-0.6) -- (9.1,-0.6);
\draw[black, thick, ->] (A.north) .. controls (1,1.7)  and (3,1.7)  .. (V.west);
\draw[black, thick, ->] (D.north) .. controls (9,1.7)  and (7,1.7)  .. (V.east);
\end{tikzpicture}
}
\end{center}

\begin{center}
\fbox{
\begin{tikzpicture}
\node[draw,rectangle,color=white,fill=white,minimum size=1.5cm] (U) at (5,-0.7) {$|{\mathbb {LOC}}\rangle$} ;
\node[draw,rectangle,color=black,fill=gray!40,minimum size=1.5cm] (V) at (5,1) { $\widehat{{\mathbb {S}}}\bigcup \widehat{{\mathbb {V}}^{\prime}}$ };
\node[builder,evil,shirt=black,hat=gray,hair=yellow,hatbadge=gray,badge=gray,minimum size=1.5cm] (A) at (0,0.1) {${S}_{1}$};
\node[builder,evil,mirrored,shirt=black,hat=gray,skin=black,hair=brown,hatbadge=gray,badge=gray,minimum size=1.5cm] (D) at (10,0.1) {${S}_{2}$};
\node[judge,mirrored,female,hair=yellow,minimum size=1.5cm] (B) at (2.5,0.1) {$J_{1}$};
\node[judge,hair=gray,minimum size=1.5cm] (C) at (7.5,0.1) {$J_{2}$};
%
\draw[yellow, ultra thick, ->] (1.91, 0.5) -- (0.57,0.5);
\draw[black, thick, ->] (1.9,0.5) -- (0.6,0.5);
\filldraw[black] (1.3,0.5) circle (0pt) node[anchor=south] {$w_{1}$};
\draw[black,fill=black!10, thick] (0.8,0) -- (0.8,-0.7) -- (1.6,-0.7) -- (1.6,0.1) -- (0.9,0.1) -- (A.mouth) -- cycle;
\draw[orange, <-] (1.5,0) -- (0.9,0);
\draw[red, ->] (1.5,-0.2) -- (0.9,-0.2);
\draw[orange, <-] (1.5,-0.4) -- (0.9,-0.4);
\draw[red, ->] (1.5,-0.6) -- (0.9,-0.6);
\draw[yellow, ultra thick, ->] (8.09, 0.5) -- (9.43,0.5);
\draw[black, thick, ->] (8.1,0.5) -- (9.4,0.5);
\filldraw[black] (8.7,0.5) circle (0pt) node[anchor=south] {$w_{2}$};
\draw[black,fill=black!10, thick] (9.2,0) -- (9.2,-0.7) -- (8.4,-0.7) -- (8.4,0.1) -- (9.1,0.1) -- (D.mouth) -- cycle;
\draw[gray, <-] (8.5,0) -- (9.1,0);
\draw[blue, ->] (8.5,-0.2) -- (9.1,-0.2);
\draw[gray, <-] (8.5,-0.4) -- (9.1,-0.4);
\draw[blue, ->] (8.5,-0.6) -- (9.1,-0.6);
\draw[black, thick, ->] (A.north) .. controls (1,1.7)  and (3,1.7)  .. (V.west);
\draw[black, thick, ->] (D.north) .. controls (9,1.7)  and (7,1.7)  .. (V.east);
\end{tikzpicture}
}
\end{center}

\begin{center}
\fbox{
\begin{tikzpicture}
\node[draw,rectangle,color=white,minimum size=1.5cm] at (5,-0.7) {};
\node[judge,female,hair=yellow,minimum size=1.5cm] (A) at (0.5,0) {$J_{1}$};
\node[judge,hair=gray,minimum size=1.5cm] (C) at (2.5,0) {$J_{2}$};
\node[judge,saturated,mirrored,minimum size=1.5cm] (B) at (5,0) {$J_{0}$};
\draw[yellow, ultra thick, ->] (4.41, 0.5) -- (0.97,0.5);
\draw[black, thick, ->] (4.4,0.5) -- (1,0.5);
\draw[yellow, ultra thick, ->] (4.41, 0.5) -- (2.97,0.5);
\draw[black, thick, ->] (4.4,0.5) -- (3.0,0.5);
\filldraw[black] (3.7,0.5) circle (0pt) node[anchor=south] {$w_{1},w_{2}$};
\draw[black,fill=black!10, thick] (1.1,0) -- (1.1,-0.7) -- (1.9,-0.7) -- (1.9,0.1) -- (1.2,0.1) -- (A.mouth) -- cycle;
\draw[orange, <-] (1.8,0) -- (1.2,0);
\draw[red, ->] (1.8,-0.2) -- (1.2,-0.2);
\draw[orange, <-] (1.8,-0.4) -- (1.2,-0.4);
\draw[red, ->] (1.8,-0.6) -- (1.2,-0.6);
%
\draw[black,fill=black!10, thick] (3.1,0) -- (3.1,-0.7) -- (3.9,-0.7) -- (3.9,0.1) -- (3.2,0.1) -- (C.mouth) -- cycle;
\draw[gray, ->] (3.8,0) -- (3.2,0);
\draw[blue, <-] (3.8,-0.2) -- (3.2,-0.2);
\draw[gray, ->] (3.8,-0.4) -- (3.2,-0.4);
\draw[blue, <-] (3.8,-0.6) -- (3.2,-0.6);
\end{tikzpicture}
}
\end{center}
\caption{\label{inter2} Interrogation or Simulation phase (top) followed by Distinction phase (bottom).}
\end{figure}
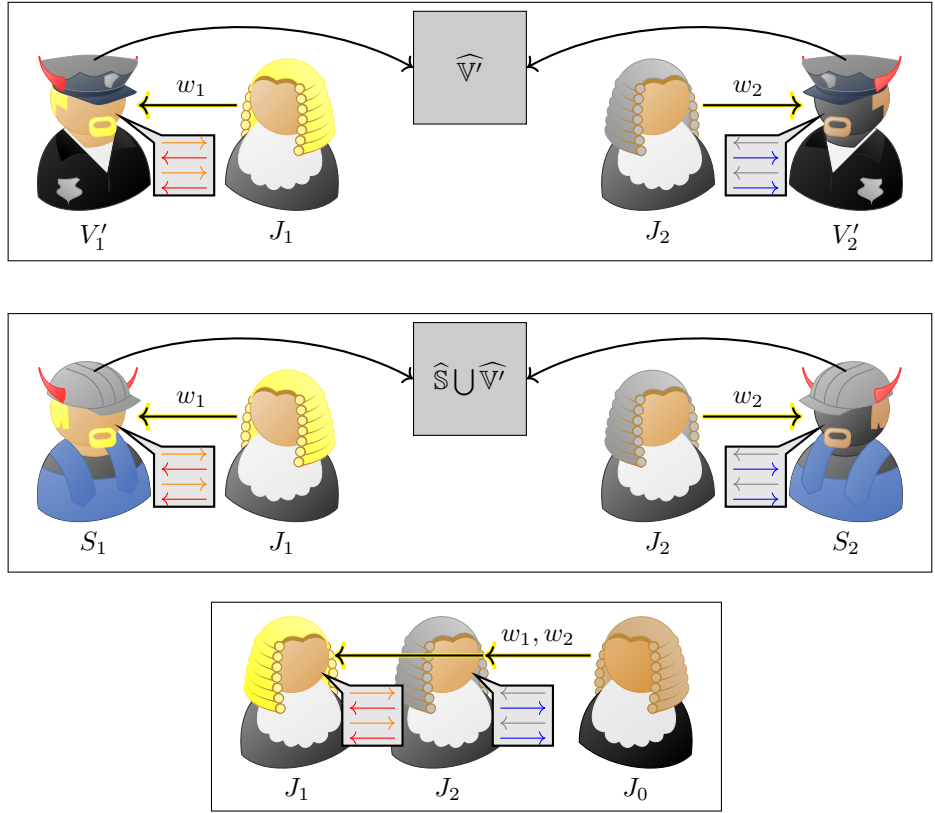

The results of this paper, depending on the specific bit commitment used, may be achieved under a stronger flavour of zero-knowledge
if a member of the non-locality class $\widehat{\mathbb {S}}$ is enough to break the binding property of the commitments. For instance, the result of section \ref{MIP_local}
is really $\mathbf{ZK}^{\overset{\text{\tiny{poly}}}{\mathbb {NOSIG}}}\!\mathbf{MIP}^{{\mathbb {LOC}}}_{\underset{\text{\tiny{poly}}}{\mathbb {LOC}}} = \mathbf{NEXP}$ although existing proofs usually mean
$\mathbf{ZK}^{\overset{\text{\tiny{poly}}}{\mathbb {SIG}}}\!\mathbf{MIP}^{{\mathbb {LOC}}}_{\underset{\text{\tiny{poly}}}{\mathbb {LOC}}} = \mathbf{NEXP}$.
Using the bit commitment scheme based on the magic square game of \cite{Crepeau2017} we can also obtain $\mathbf{ZK}^{\overset{\text{\tiny{poly}}}{\mathbb {|LOC\rangle}}}\!\mathbf{MIP}^{{\mathbb {LOC}}}_{\underset{\text{\tiny{poly}}}{\mathbb {LOC}}} = \mathbf{NEXP}$.

Some interesting questions resulting from this definition is whether any higher class such as
$\mathbf{ZK}^{\overset{\text{\tiny{poly}}}{\mathbb {LOC}}}\!\mathbf{MIP}^{{\mathbb {LOC}}}_{\underset{\text{\tiny{poly}}}{\mathbb {LOC}}}$ or 
$\mathbf{ZK}^{\overset{\text{\tiny{poly}}}{\mathbb {NOSIG}}}\!\mathbf{MIP}^{{\mathbb {NOSIG}}}_{\underset{\text{\tiny{poly}}}{\mathbb {NOSIG}}}$ contains more than the natural examples
such as GRAPH ISO or CODE EQUIV already found in the most natural class 
$\mathbf{ZK}^{\overset{\text{\tiny{poly}}}{\mathbb {SIG}}}\!\mathbf{MIP}^{{\mathbb {SIG}}}_{\underset{\text{\tiny{poly}}}{\mathbb {SIG}}} = \mathbf{ZK}\!\mathbf{IP}$.

\subsection{A note on notation}
$$\mathlarger{\mathlarger{\mathlarger{\mathbf{ZK}^{{\mathbb {S}}}\!\!\;\mathbf{MIP}^{{\mathbb {P}}}_{{\mathbb {V}}}}}}$$ is the complexity class of Zero-Knowledge Multi-provers Interactive Proofs where (honest and dishonest) provers are restricted to non-locality class $\mathbb {P}$ (important for soundness),
where the honest verifier is from non-locality class $\mathbb {V}$ (also important for soundness), and where the Zero-Knowledge simulators are from non-locality class $\mathbb {S}$ unless $\widehat{V^\prime}$ is outside of $\mathbb {S}$ in which case they are from the class of $\widehat{V^\prime}$.


\end{document}